\documentclass[11pt]{article}

\usepackage[utf8]{inputenc}
\pdfoutput=1
\usepackage[margin=1.1in,letterpaper]{geometry}

\usepackage{bm}

\makeatletter
\newlength\aftertitskip     \newlength\beforetitskip
\newlength\interauthorskip  \newlength\aftermaketitskip

\makeatother

\pdfoutput=1                    % force pdflatex to run

\usepackage{amsmath,amsthm}
\usepackage{cite}
\usepackage{graphicx}
\usepackage{color}
\usepackage{amssymb,amsfonts,amsxtra} %,mathrsfs,bm}
\usepackage{mathtools}
\usepackage{multicol}
\usepackage{multirow}
\usepackage{paralist}
\usepackage{physics}
\usepackage{xspace}
\usepackage{algorithm}% http://ctan.org/pkg/algorithm
\usepackage{algpseudocode}
\usepackage[colorlinks=true, urlcolor = blue, citecolor=black]{hyperref}
\usepackage{bbm}
\usepackage{nicefrac}
\usepackage{enumitem}

\newtheorem{theorem}{Theorem}[section]
\newtheorem{lemma}[theorem]{Lemma}
\newtheorem{corollary}[theorem]{Corollary}
\newtheorem{proposition}[theorem]{Proposition}

\theoremstyle{definition}
\newtheorem{definition}[theorem]{Definition}
\newtheorem{remark}[theorem]{Remark}

\numberwithin{equation}{section}

\title{Spreading and Structural Balance on Signed Networks}
\author{Yu Tian\thanks{Nordita, Stockholm University and KTH Royal Institute of Technology, SE-106 91 Stockholm, Sweden
	(yu.tian@su.se).}
\and Renaud Lambiotte\thanks{Mathematical Institute, University of Oxford, OX2 6GG Oxford, UK
	(renaud.lambiotte@maths.ox.ac.uk).}}
	\date{}
	
\begin{document}
\maketitle

% REQUIRED < 250 words
\begin{abstract}
	Two competing types of interactions often play an important part in shaping system behavior, such as activatory or inhibitory functions in biological systems. Hence, signed networks, where each connection can be either positive or negative, have become popular models over recent years. However, the primary focus of the literature is on the unweighted and structurally unbalanced ones, where all cycles have an even number of negative edges. Hence here, we first introduce a classification of signed networks into balanced, antibalanced or strictly unbalanced ones, and then characterize each type of signed networks in terms of the spectral properties of the signed weighted adjacency matrix. In particular, we show that the spectral radius of the matrix with signs is smaller than that without if and only if the signed network is strictly unbalanced. These properties are important to understand the dynamics on signed networks, both linear and nonlinear ones. Specifically, we find consistent patterns in a linear and a nonlinear dynamics theoretically, depending on their type of balance. We also propose two measures to further characterize strictly unbalanced networks, motivated by perturbation theory. Finally, we numerically verify these properties through experiments on both synthetic and real networks. 
	%The results can also contribute to a better understanding of the problems that are closely related to the dynamics, e.g., the influence maximization.
\end{abstract}

% REQUIRED
% \begin{keywords}
	% Signed networks, structural balance and antibalance, weighted adjacency matrix, linear and nonlinear dynamics, random walks
	% \end{keywords}

% REQUIRED
% \begin{MSCcodes}
	% 05C22, 05C50, 05C81, 37E25, 39A06, 91D30, 94C15
	% \end{MSCcodes}
% 05C22 - signed and weighted graphs
% 05C50 - graphs and linear algebra
% 05C81 - random walks on graphs
% 37E25 - dynamical systems involving maps of trees and graphs
% 39A06 - linear difference equations
% 82B20 - lattice systems (Ising, dimer, Potts, etc.) and systems on graphs arising in equilibrium statistical mechanics
% 82C20 - dynamic lattice systems (kinetic Ising, etc.) and systems on graphs in time-dependent statistical mechanics
% 91D30 - social networks; opinion dynamics
% 92E10 - molecular structure (graph-theoretic methods, methods of differential topology, etc.)
% 94C15 - applications of graph theory to circuits and networks

\section{Introduction}
% \begin{itemize}
	%     \item Signed networks, together with dynamics on them, are interesting and find many applications.
	%     \item Structural balance is a good way to consider signed networks. 
	%     \item The development and our contributions. 
	%     \item Outline.
	% \end{itemize}
% dynamics on networks -> signed networks
The study of dynamics on networks has attracted much research interest recently due to its applications in engineering, physics, biology, and social sciences, e.g., \cite{Couzin_2005_animal,Jadbabaie_2003_coordination,lambiotte_2022_dynamics,porter_2016_dynamics}. In particular, one can distinguish two general classes of models: linear models and nonlinear models. With roots traceable back to topics such as the ``Gambler’s ruin” problem \cite{Ore_1960_pascal}, the spreading of disease \cite{Mollison_1977_epidemic} and random-walk processes on networks \cite{Masuda2017RW}, linear dynamical processes have been a popular class of models to understand diffusion in various contexts. Meanwhile, nonlinear models have also been analyzed extensively to incorporate more complexity of the dynamical processes \cite{Guibeault_Complcontag_2018,Srivastava_bifur_2010}. A fundamental idea in this area is that by characterizing the underlying network structure between agents, collective dynamics can be predicted or controlled in a systematic manner. Several important results have been established, e.g., that the modular structure can effectively simplify the description of dynamical systems \cite{lambiotte_2022_dynamics}.

Simple networks where there is only one single type of connection can describe a system reasonably well in many cases, but the coexistence of two competing types of interactions can become essential to shape the system behavior, e.g., activatory or inhibitory functions in biological systems, trustful or mistrustful connections in social or political networks, and cooperative or antagonistic relationships in the economic world \cite{Altafini_2012_opinion,Facchetti_2011_large,Marvel_2010_signedContDyn,tian_2022_thesis,tian_2021_role}. Therefore, \textit{signed networks}, where connections can be either positive or negative, have become important ingredients of models in many research fields over recent years. Also in mathematics, signed networks play an important role in various branches, such as group theory, topology and mathematical physics \cite{Bieche_1980_spin,Bilu_2006_spectGap,Cameron_1977_cohomological,Cameron_1994_groups,Cameron_1986_classes}.

% structural balanced and related results.
A central notion in the study of signed networks is that of structural balance \cite{Easley_2010_networks,Facchetti_2011_large,Kunegis_2009_Zoo,SzellEtal_2010_multirelational}. This concept has initially been motivated by problems in social psychology \cite{cartwrightharary_1956_gbalance,harary_1953_balance} and has stimulated new methods for analyzing social networks \cite{Kunegis_2009_Zoo,Symeonidis_2013_multiway,Wu_2012_eigenspace}, biological networks \cite{Symeonidis_2013_biology} and so on. In particular, a signed network is structurally balanced if and only if all its cycles are so-called positive, which can be characterized in terms of the smallest eigenvalue of the signed (normalized) Laplacian \cite{Kunegis_signspect_2010,zaslavsky_1982_signed}. Researchers have shown that in the case of structurally balanced networks, the behavior of the dynamics is largely predictable, and can resort to the corresponding dynamical systems theory, such as consensus dynamics \cite{Altafini_2012_opinion}. 

% the state of the art and where we fit
However, the dynamical properties when the underlying signed networks are not structurally balanced are relatively unknown. In addition, when considering structural properties, a majority of works focus on unweighted signed networks. For this reason, in this paper, we consider dynamics on signed networks where edges can be weighted, and investigate the whole range of situations when the signed network may be balanced, antibalanced, and strictly unbalanced. Our first contribution is to characterize each type of signed networks in terms of the spectral properties of the signed weighted adjacency matrix, and in particular, we show that the spectral radius of the matrix with signs is smaller than that without if and only if the signed network is strictly unbalanced. Then, we exploit this result to understand both linear and nonlinear dynamics on networks, through a linear dynamics model where the coupling matrix is the weighted adjacency matrix (``linear adjacency dynamics" hereafter) and the extended linear threshold model \cite{Tian_info_2021}, appropriately generalized to signed networks in this paper. The two examples are important models in various contexts, e.g., in information propagation. Our second contribution is to show consistent patterns of these two separate dynamics on signed networks depending on their type of balance. We also propose two measures to further characterize strictly unbalanced networks, motivated by perturbation theory. Finally, the results are numerically verified in both synthetic and real networks.

% The outline
This paper is organized as follows. In section \ref{sec:preliminary}, we review the important concepts in signed networks, including the signed Laplacians, structural balance, and two basic rules in defining dynamics on signed networks. Specifically, we explain in detail how to extend random walks to signed networks. The main results are covered in section \ref{sec:main_results}. Specifically, in subsection \ref{sec:structprop}, we first discuss the classification of signed networks in subsection \ref{sec:struct-classif}, and then characterize the spectral properties in each type in subsection \ref{sec:struct-spectrum}. Based on the understanding of the structure, in subsection \ref{sec:dynamicprop}, we further characterize the dynamical properties in terms of both the linear adjacency dynamics in subsection \ref{sec:linear_adj_dynamics} and the extended linear threshold model in subsection \ref{sec:elt}. With the help of the classification, these two different classes of models can exhibit similar performances. As a slightly modified example of the linear adjacency dynamics, we also characterize both the short-term and long-term behavior of signed random walks in subsection \ref{sec:random_walk}. Finally, we verify the results on both synthetic networks that are balanced, antibalanced and strictly unbalanced, obtained from signed stochastic block models, and a real example of signed networks that are strictly unbalanced in section \ref{sec:numerical_exp}. In section \ref{sec:conclusions}, we conclude with potential future directions.

\section{Preliminary}
\label{sec:preliminary}
In this section, we introduce some mathematical preliminaries on signed networks, signed Laplacians, structural balance, and the dynamics on signed networks. Specifically, we illustrate in detail how to extend random walks to signed networks, which we consider as an example of the linear (adjacency) dynamics later in subsection \ref{sec:random_walk}.
% \begin{itemize}
	%     \item The definition of signed networks, and Laplacians.
	%     \item The definitions of structural balance, and the results of signed Laplacian and signed normalised Laplacian. 
	%     \item The dynamics on signed networks: the opposing rule, as well as the repelling rule. \item The dynamics we are interested in: linear adjacency dynamics and the extended threshold model on signed networks. 
	% \end{itemize}
\subsection{Signed networks}
We take $G = (V, E, \mathbf{W})$ as an undirected signed network, where $V = \{v_1, v_2, \dots, v_n\}$ is the node set, an edge $(v_i, v_j)\in E$ is an unordered pair of two distinct nodes in the set $V$, and the signed weighted adjacency matrix $\mathbf{W}\in \mathbb{R}^{n\times n}$ describes the nonzero edge weights. Each edge in $E$ is associated with a sign, positive or negative, characterizing $G$ as a signed network. Specifically, if there is no edge between nodes $v_i, v_j$, $W_{ij} = 0$; otherwise, $W_{ij} > 0$ denotes a positive edge, while $W_{ij} < 0$ denotes a negative edge. The degree of a node $v_i$ is defined as
\begin{equation}
	d_i = \sum_j\abs{W_{ij}},
	\label{equ:signed_d}
\end{equation}
and, motivated by graph drawing \cite{Kunegis_signspect_2010}, the signed Laplacian matrix in the literature is normally defined as 
\begin{equation}
	\mathbf{L} = \mathbf{D} - \mathbf{W},
	\label{equ:signed_L}
\end{equation}
where the signed degree matrix $\mathbf{D}$ is the diagonal matrix with $\mathbf{d} = (d_i)$ on its diagonal. Accordingly, the signed random walk Laplacian is defined as
\begin{equation}
	\mathbf{L}_{rw} = \mathbf{I} - \mathbf{D}^{-1}\mathbf{W},
	\label{equ:signed_Lrw}
\end{equation}
where $\mathbf{I}$ is the identity matrix, for reasons that will be clarified in section \ref{sec_rw}. Most work in the literature is based on unweighted signed networks, hence in this section, we assume $G$ to be unweighted unless otherwise explicitly mentioned, i.e., $\mathbf{W} = \mathbf{A}$ where $\mathbf{A} = (A_{ij})$ is the (unweighted) adjacency matrix with
\begin{equation}
	A_{ij} = 
	\begin{cases}
		sign(W_{ij}),\quad &\text{if } W_{ij}\ne 0;\\
		0,\quad &\text{otherwise},
	\end{cases}
	\label{equ:signed_A}
\end{equation}
and function $sign(\cdot): \mathbb{R} \to \{-1,0,1\}$ indicates the sign of a value.

\subsection{Structural balance and antibalance}
Introduced in 1940s \cite{heider_1946_psychology} and primarily motivated by social and economic networks, a fundamental notion in the study of signed networks is the so-called \textit{structural balance} \cite{cartwrightharary_1956_gbalance}. A signed graph is structurally balanced if and only if there is no cycle with an odd number of negative edges, which defines the cycle to be ``negative". The following theorem provides an alternative interpretation of structural balance in terms of a bipartition of signed graphs.
\begin{theorem}[structure theorem for balance \cite{harary_1953_balance}]
	A signed graph $G$ is structurally balanced if and only if there is a bipartition of the node set into $V=V_1\cup V_2$ with $V_1$ and $V_2$ being mutually disjoint and one of them being nonempty, s.t.~any edge between the two node subsets is negative while any edge within each node subset is positive. 
	\label{the:lit_signed-bal}
\end{theorem}
%If $G$ is a complete graph, it turns out that we can verify its structural balance property by simply checking all triangles: $G$ is structurally balanced if and only if among every set of three nodes there are either one or three positive edges [14]. 

From another aspect, Harary \cite{harary_1957_duality} defined a signed graph $G$ to be antibalanced if the graph negating the edge sign is balanced. Thus, $G$ is antibalanced if and only if there is no cycle with an odd number of positive edges. By reversing the edge sign, Harary gave the following antithetical dual result for antibalance \cite{harary_1957_duality}.
\begin{theorem}[structure theorem for antibalance \cite{harary_1957_duality}]
	A signed graph $G$ is structurally antibalanced if and only if there is a bipartition of the node set into $V=V_1\cup V_2$ with $V_1$ and $V_2$ being mutually disjoint and one of them being nonempty, s.t.~any edge between the two node subsets is positive while any edge within each node subset is negative. 
	\label{the:lit_signed-anti}
\end{theorem}
There is a further line of research in the weakened version of structural balance, where a graph is weakly balanced if and only if no cycle has exactly one negative edge in $G$ \cite{davis_1967_balance,Easley_2010_networks}. But due to its lack of dynamical interpretation, we focus on the original version of structural balance in this paper.

The properties of being balanced or antibalanced can be characterized by the eigenvalues of both the signed Laplacian matrix \eqref{equ:signed_L} and the signed random walk Laplacian \eqref{equ:signed_Lrw}. Specifically, Kunegis \textit{et al.}~showed that, as in the case of the unsigned Laplacian, the signed Laplacian matrix is still positive semi-definite, and it is positive definite if and only if the underlying signed network does not have a balanced connected component \cite{Kunegis_signspect_2010}. Similarly, the smallest eigenvalue of the signed random walk Laplacian vanishes if and only if the network has a balanced connected component. Meanwhile, it is also known that a signed network has an antibalanced connected component if and only if the largest eigenvalue of the signed random walk Laplacian equals $2$ \cite{Li_2009_normL}. The counterpart for the signed Laplacian has also been explored, and we refer the reader to \cite{hou_2003_Laplacian,zaslavsky_1982_signed} for more details in the results of the spectral properties of signed networks. 

The literature investigating the properties of  signed networks that are neither balanced nor antibalanced is more limited. Among them, Atay and Liu characterized such signed networks through the idea of Cheeger inequality \cite{Atay_signedCheeger_2020}. They defined the signed Cheeger constant through how far the network is from having a balanced connected component, and managed to estimate the smallest eigenvalue of the Laplacian matrices from below and above with it. They obtained similar results concerning antibalance and the spectral gap between $2$ and the largest eigenvalue, via an antithetical dual signed Cheeger constant. We refer the reader to \cite{Belardo_2014_Laplacian,Hou_2005_Laplacian} for more work on this aspect.

\subsection{Dynamics on signed networks}
\label{sec_rw}
There are various research directions when analyzing dynamics over signed networks, and definitions of models typically differ in how to interpret the different influence played by a positive versus a negative edge on the dynamics. For example, different consensus algorithms with positive and negative edges have been proposed and investigated \cite{Altafini_consensus_2013,Hendrickx_opinion_2014,Liu_convergence_2017,Meng_multiagen_2015,Meng_antagonistic_2016,Shi_disagree_2013,Shi_beliefs_2016,Xia_balanceDynamics_2016}. There exist two basic types of interactions along the negative edges: the ``opposing negative dynamics" \cite{Altafini_consensus_2013} where nodes are attracted by the opposite values of the neighbours, and the ``repelling negative dynamics" \cite{Shi_disagree_2013} where nodes tend to be repulsive of the relative value of the states with respect to the neighbours instead of being attractive. We refer the reader to Shi, Altafini and Baras \cite{Shi_dynamicsSigned_2019} for a recent review of dynamics on signed networks through extending the classic DeGroot model with the two aforementioned rules on negative edges. In this section, we present in detail the opposing negative dynamics, due to its connection to signed Laplacian and, as we will see, the notion of structural balance in signed networks.

The dynamics on signed networks induced by the opposing rule play an important part in various contexts \cite{Altafini_consensus_2013,Proskurnikov_2018_dynamic,Valcher_2014_bipartite}. Here, we give an interpretation in terms of \textit{random walks} on signed networks. We first consider an unweighted case and write the (unweighted) adjacency matrix as $\mathbf{A} = \mathbf{A}^+ - \mathbf{A}^-$, where $A^+_{ij} = 1$ if there is a positive edge between nodes $v_i$ and $v_j$ and $A^-_{ij} = 1$ if there is a negative edge between them, and we also represent the degree of each node $v_j$ as $d_j = d^+_j + d^-_j$, where $d^+_j$ is the number of positive neighbours of $v_j$ and $d^-_j$ is the number of negative neighbours of $v_j$. In addition, we assume that there are two types of walkers, positive and negative walkers, whose densities on node $v_i$ are $x_i^+$ and $x_i^-$, respectively. Furthermore, guided by the opposing rule, we define that negative edges can flip the sign of walkers going through the edges, while their sign remains unchanged when going through positive edges. That is, a positive walker becomes negative after traversing a negative edge, while it conserves its sign while traversing a positive edge, for instance.
Hence, on each node, there are two different sources for positive walkers, either from positive walkers through positive edges or from negative walkers through negative edges, i.e., 
\begin{eqnarray}
	\label{sys1}
	x_j^+(t) = \sum_i\frac{1}{d_i}\left(A^+_{ij}x_i^+(t-1) + A^-_{ij}x_i^-(t-1)\right); 
\end{eqnarray}
there are also two different sources for negative walkers on each node, either from positive walkers through negative edges or from negative walkers through positive edges, i.e., 
\begin{eqnarray}
	\label{sys2}
	x_j^-(t) = \sum_i\frac{1}{d_i}\left(A^-_{ij}x_i^+(t-1) + A^+_{ij}x_i^-(t-1)\right).
\end{eqnarray}
The whole dynamics is thus governed by the transition matrix of a $2n \times 2n$ matrix, which can be interpreted as the adjacency matrix of a larger graph, where each node appears twice,
\begin{eqnarray}
	\mathbf{A}^{(2)}=\begin{pmatrix}
		\mathbf{A}^+  & \mathbf{A}^-\\
		\mathbf{A}^- & \mathbf{A}^+ 
	\end{pmatrix}.
	\label{equ:A2}
\end{eqnarray}
Indeed, after noting that $\sum_j A^{(2)}_{ij}=d_i $, the system characterized by \eqref{sys1} and \eqref{sys2} has the coupling matrix $\mathbf{P}^{(2)} = \mathbf{D}^{(2)-1}\mathbf{A}^{(2)}$, where $\mathbf{D}^{(2)} = [\mathbf{D}, \mathbf{0}; \mathbf{0}, \mathbf{D}]$ is the diagonal matrix with $\mathbf{d} = (d_j)$ on the diagonal but the appearance doubled. Note that a similar matrix has been investigated in the context of spectral clustering, with sterling results, but it was introduced from a different perspective, via a Gremban's expansion of a Laplacian system \cite{fox2018investigation}.

The properties of the system of equations is thus governed by the spectral properties of $\mathbf{P}^{(2)}$, which can be obtained from two well-known matrices as follows. The total number $n_i$ of walkers on node $v_i$ is obtained by taking the sum of \eqref{sys1} and \eqref{sys2}, leading to
\begin{align*}
	n_j(t) = x_j^+(t) + x_j^-(t)
	= \sum_i\frac{1}{d_i}\left(A^+_{ij} + A^-_{ij}\right) n_i(t) ,  
\end{align*}
where $\frac{1}{d_i}\left(A^+_{ij} + A^-_{ij}\right)$ is the classical transition matrix of the unsigned network whose edge signs are not considered, as expected.

In contrast, the ``polarization" on each node, defined as the difference between the number of positive and negative walkers $x_j(t) = x_j^+(t) - x_j^-(t)$, is obtained by subtracting Equation \eqref{sys2} from \eqref{sys1}, giving
\begin{align*}
	x_j(t) = x_j^+(t) - x_j^-(t) &= \sum_i\frac{1}{d_i}\left(\left(A^+_{ij} - A^-_{ij}\right)x_i^+(t-1) - \left(A^+_{ij} - A^-_{ij}\right)x_i^-(t-1)\right)\\
	&= \sum_i\frac{1}{d_i}\left(A^+_{ij} - A^-_{ij}\right)\left(x_i^+(t-1) - x_i^-(t-1)\right) = \sum_i\frac{1}{d_i}A_{ij}x_i(t-1).  
\end{align*}
This gives the signed (unweighted) transition matrix $\mathbf{P} = \mathbf{D}^{-1}\mathbf{A}$, and the weighted version in terms of $\mathbf{W}$ can be obtained similarly. These equations can then be extended from a discrete-time setting to a continuous-time setting classically \cite{Masuda2017RW}, by assuming that walkers jump at continuous rate or at a rate proportional to the node degree, leading to the signed random walk Laplacian as in \eqref{equ:signed_Lrw} and accordingly the signed Laplacian as in \eqref{equ:signed_L}, respectively. 

Related to random walk processes on networks, the problem of information propagation on networks with only positive connections has been studied extensively in the literature \cite{Mossel_Roch_2010,Pastor-Satorras_epidemic_2015,pastor-satorras_internet_2004}, but relatively less has been explored in the context of signed networks. Among such work, Li \textit{et al.}~\cite{Li_voterIM_2015} extended the voter model to signed networks with the ``opposing negative dynamics", and we also refer the reader to \cite{Chen_signedIC_2011,Liu_signedIC_2019} for extending the classic independent cascade model to signed networks. 

In this paper, we consider two dynamics that are closely related to the information propagation but with more general properties, the linear adjacency dynamics and the extended linear threshold model as defined in \cite{Tian_info_2021}, both of which have been extended to signed networks with the opposing rule. Note that here we consider dynamics on a fixed signed network. There is another line of research to explore the evolution of edge weights in signed networks and their interactions with the balanced structure. However, it is out of scope of our current analysis, and we refer the reader to \cite{Marvel_2010_signedContDyn} and references therein.

\section{Main results}
\label{sec:main_results}
In this section, we further present the interesting properties of signed networks that we found, where we will show how a signed network connects and differentiates from its unsigned counterpart from both the structural and the dynamical perspectives, and how the separate behavior interacts with the structure balance. Throughout the section, we consider connected\footnote{\footnotesize{For disconnected networks, we can consider each of its connected components.}}, undirected and weighted signed networks $G = (V, E, \mathbf{W})$ where $\mathbf{W}$ is the signed weighted adjacency matrix, and the corresponding networks ignoring the edge sign $\bar{G} = (V, E, \bar{\mathbf{W}})$ where $\bar{\mathbf{W}}$ is the unsigned weighted adjacency matrix with $\bar{W}_{ij} = \abs{W_{ij}},\ \forall v_i,v_j\in V$.

\subsection{Structural properties} 
\label{sec:structprop}
We start from the structural properties characterized by the signed weighted adjacency matrix $\mathbf{W}$. Specifically, based on the analysis of structurally balanced and antibalanced graphs, we first introduce the classification of signed networks we will follow throughout this paper in subsection \ref{sec:struct-classif}. Then with the classification, we illustrate how the structural properties in each category different from each other through the spectrum of  $\mathbf{W}$ in subsection \ref{sec:struct-spectrum}. 
% \begin{itemize}
	%     \item Structural balances and related properties.
	%     \item Spectral decomposition and the spectral radius of strictly unbalanced graphs. 
	% \end{itemize}

\subsubsection{Classifications}
\label{sec:struct-classif}
We define a signed graph to be \textit{balanced} as in Theorem \ref{the:lit_signed-bal}, and \textit{antibalanced} as in Theorem \ref{the:lit_signed-anti}. Finally, we define all the remaining signed graphs to be \textit{strictly unbalanced} in Definition \ref{def:signed_strictunb}. Since we focus on the structural properties here, we use the terms ``network" and ``graph" interchangeably. 
\begin{definition}[strict unbalance]
	A signed graph $G$ is strictly unbalanced if $G$ is neither balanced nor antibalanced. 
	\label{def:signed_strictunb}
\end{definition}

With the definitions, we should note that balanced graphs and antibalanced graphs are not always mutually exclusive. For example, a four-node path with edge sign $(-, +, -)$ is both balanced and antibalanced, and so is a four-node cycle with edge sign $(-, +, -, +)$. Indeed, the intersection between balanced graphs and antibalanced graphs only contains signed trees and balanced/antibalanced bipartite graphs, as illustrated in Propositions \ref{pro:balanced-anti-tree} and \ref{pro:balanced-anti-bipart}. As in the literature, we define a path, walk, or cycle to be \textit{positive} if it contains an even number of negative edges, and \textit{negative} otherwise. 
\begin{proposition}
	Every signed tree is both balanced and antibalanced. 
	\label{pro:balanced-anti-tree}
\end{proposition}
\begin{proof}
	This follows from the definitions. See also Appendix \ref{sec:app_proofs} for a constructive proof with the corresponding bipartitions found, and Lemma 3.1 in \cite{zaslavsky_1982_signed}.
\end{proof}

\begin{proposition}
	A (non-tree) balanced signed graph $G$ is antibalanced if and only if it is bipartite. 
	\label{pro:balanced-anti-bipart}
\end{proposition}
\begin{proof}
	This can be derived from the definitions, and see Appendix \ref{sec:app_proofs} for more detail.
\end{proof}

\subsubsection{Spectrum}
\label{sec:struct-spectrum}
With a better understanding of different types of signed networks, we now characterize them through their spectral properties. Specifically, we show that the eigenvalues and eigenvectors of a signed graph $G$ are closely related to those of its unsigned counterpart $\bar{G}$ in Theorem \ref{the:transition-spect}, and further characterize the leading eigenvalue and eigenvector in Proposition \ref{pro:transition-spect-rho}. Similar results have been considered in the literature but the primary focus is on unweighted graphs or other matrices, e.g., the signed Laplacian \cite{Kunegis_signspect_2010}. 
\begin{theorem}[spectral theorem of balance and antibalance]
	Let $\mathbf{W} = \mathbf{U}\Lambda\mathbf{U}^{T}$ and $\bar{\mathbf{W}} = \bar{\mathbf{U}}\bar{\Lambda}\bar{\mathbf{U}}^{T}$ be the unitary eigendecompositions of $\mathbf{W}$ and $\bar{\mathbf{W}}$, respectively, where $\mathbf{U}\mathbf{U}^T = \mathbf{I}$ and $\bar{\mathbf{U}}\bar{\mathbf{U}}^T = \mathbf{I}$. Let $V_1,V_2$ denote the corresponding bipartition for either balanced or antibalanced graphs, and $\mathbf{S}$ denote the diagonal matrix whose $(i,i)$ element is $1$ if $i\in V_1$ and $-1$ otherwise. 
	\begin{enumerate}
		\item If $G$ is balanced, 
		\begin{align*}
			\Lambda = \bar{\Lambda},\ \mathbf{U} = \mathbf{S}\bar{\mathbf{U}}.
		\end{align*}
		\item If $G$ is antibalanced,
		\begin{align*}
			\Lambda = -\bar{\Lambda},\ \mathbf{U} = \mathbf{S}\bar{\mathbf{U}}.
		\end{align*}
	\end{enumerate}
	\label{the:transition-spect}
\end{theorem}
\begin{proof}
	If $G$ is balanced, $\mathbf{W} = \mathbf{S}\bar{\mathbf{W}}\mathbf{S}$ by definition. Then, 
	\begin{align*}
		\mathbf{W}= \mathbf{S}\bar{\mathbf{W}}\mathbf{S} = \mathbf{S}\bar{\mathbf{U}}\bar{\Lambda}\bar{\mathbf{U}}^{T}\mathbf{S} = (\mathbf{S}\bar{\mathbf{U}})\bar{\Lambda
		}(\mathbf{S}\bar{\mathbf{U}})^{T},
	\end{align*}
	where the third equality is by $\mathbf{S}\mathbf{S} = \mathbf{I}$. It is the eignedecomposition of $\mathbf{W}$ by the uniqueness. Hence, $\Lambda = \bar{\Lambda},\ \mathbf{U} = \mathbf{S}\bar{\mathbf{U}}$.  
	
	While, if $G$ is antibalanced, $\mathbf{W} = -\mathbf{S}\bar{\mathbf{W}}\mathbf{S}$ by definition. Then, 
	\begin{align*}
		\mathbf{W}= -\mathbf{S}\bar{\mathbf{W}}\mathbf{S} = -\mathbf{S}\bar{\mathbf{U}}\bar{\Lambda}\bar{\mathbf{U}}^{T}\mathbf{S} = (\mathbf{S}\bar{\mathbf{U}})(-\bar{\Lambda
		})(\mathbf{S}\bar{\mathbf{U}})^{T},
	\end{align*}
	where the third equality is by $\mathbf{S}\mathbf{S} = \mathbf{I}$. It is the eignedecomposition of $\mathbf{W}$ by the uniqueness. Hence, $\Lambda = -\bar{\Lambda},\ \mathbf{U} = \mathbf{S}\bar{\mathbf{U}}$.   
\end{proof}
We also note that the above results can be derived from the spectral invariance property of a graph theoretical concept of \textit{switching equivalence}, and we refer the reader to \cite{zaslavsky_1982_signed} for more results from this perspective.
\begin{remark}
	For directed signed graphs, we can show that (i) the relationships between the eigenvalues still hold, and (ii) the general eigenvectors of the two matrices have the same correspondence as the eigenvectors in Theorem \ref{the:transition-spect}, where the proof follows similarly but replacing the unitary decomposition by their Jordan canonical forms. We can also show similar relationships in terms of singular values, and left-singular and right-singular vectors by considering their singular value decomposition instead of the eigendecomposition.
\end{remark}
\begin{proposition}
	Suppose $\bar{G}$ is not bipartite, or is aperiodic. Let $\lambda_1\ge \lambda_2 \ge \dots \ge \lambda_n$ denote the eigenvalues of $\mathbf{W}$ with the associated eigenvectors $\mathbf{u}_1, \mathbf{u}_2, \dots, \mathbf{u}_n$, $\bar{\lambda}_1\ge \bar{\lambda}_2 \ge \dots \ge \bar{\lambda}_n$ denote the eigenvalues of $\bar{\mathbf{W}}$ with the associated eigenvectors $\bar{\mathbf{u}}_1, \bar{\mathbf{u}}_2, \dots, \bar{\mathbf{u}}_n$, and $\rho(\cdot)$ denotes the spectral radius. Let $V_1, V_2$ denote the corresponding bipartition for either balanced or antibalanced graphs. 
	\begin{enumerate}
		\item If $G$ is balanced, $\lambda_1 = \rho(\mathbf{W}) > 0$, and this eigenvalue is simple and the only one of the largest magnitude, where $\abs{\lambda_i} < \lambda_1,\ \forall i\ne 1$.
		\item If $G$ is antibalanced, $\lambda_n = -\rho(\mathbf{W}) < 0$, and this eigenvalue is simple and the only one of the largest magnitude, where $\abs{\lambda_i} < -\lambda_n,\ \forall i\ne n$.
	\end{enumerate}
	Meanwhile, the associated eigenvector, $\mathbf{u}_1$ for balanced graphs and $\mathbf{u}_n$ for antibalanced graphs, is the only one of the following pattern: it has positive values in one node subset in the bipartition (e.g., $V_1$) and negative values in the other (e.g., $V_2$).
	\label{pro:transition-spect-rho}
\end{proposition}
\begin{proof}
	Since $\bar{\mathbf{W}}$ is an non-negative matrix, and $\bar{G}$ is irreducible and aperiodic, then by Perron-Frobenius theorem, (i) $\rho(\bar{\mathbf{W}})$ is real positive and an eigenvalue of $\bar{\mathbf{W}}$, i.e., $\bar{\lambda}_1 = \rho(\bar{\mathbf{W}})$, (ii) this eigenvalue is simple s.t.~the associated eigenspace is one-dimensional, (iii) the associated eigenvector, i.e., $\bar{\mathbf{u}}_1$, has all positive entries and is the only one of this pattern, and (iv) $\bar{\mathbf{W}}$ has only $1$ eigenvalue of the magnitude $\rho(\bar{\mathbf{W}})$. 
	
	Then, if $G$ is balanced, from Theorem \ref{the:transition-spect}, (i) $\mathbf{W}$ and $\bar{\mathbf{W}}$ share the same spectrum, and (ii) $\mathbf{U} = \mathbf{S}\bar{\mathbf{U}}$, where $\mathbf{U} = [\mathbf{u}_1, \mathbf{u}_2, \dots, \mathbf{u}_n]$ and $\bar{\mathbf{U}} = [\bar{\mathbf{u}}_1, \bar{\mathbf{u}}_2, \dots, \bar{\mathbf{u}}_n]$ containing all the eigenvectors, and $\mathbf{S}$ is the diagonal matrix whose $(i,i)$ element is $1$ if $v_i\in V_1$ and $-1$ otherwise. Hence, $\lambda_1 = \bar{\lambda}_1 = \rho(\bar{\mathbf{W}}) = \rho(\mathbf{W})$, and this eigenvalue is simple and the only one of the largest magnitude. Meanwhile, $\mathbf{u}_1 = \mathbf{S}\bar{\mathbf{u}}_1$, thus it has the pattern as described and is the only one of this pattern. The results of antibalanced graphs follow similarly. 
\end{proof}
\begin{remark}
	For a bipartite undirected graph $G$, if it is balanced, it will also be antibalanced, and similar results follow except that the corresponding eigenvalue is the only one of the largest magnitude (see Appendix \ref{sec:app_spect-bi} for details). 
\end{remark}

Finally, we consider strictly unbalanced graphs. It is a relatively unexplored area, and the existing results are mostly with regard to the signed Laplacian matrices \cite{Atay_signedCheeger_2020,Kunegis_signspect_2010}. Here, we show a general property in terms of the weighted adjacency matrix that its spectral radius is smaller than the unsigned counterpart given that the signed network is neither balanced nor antibalanced. Together with Theorem \ref{the:transition-spect}, this is the only case when the contraction of the spectral radius occurs. Hence, if we consider dynamics given by the signed weighted adjacency matrix, the corresponding state values obtained from strictly unbalanced graphs will generally have smaller magnitude (at least in long term) compared with those obtained from either balanced or antibalanced graphs (subject to appropriate initialization).
\begin{lemma}
	If $G$ is strictly unbalanced, then $\exists v_i,v_j\in V$ and $l\in \mathbb{Z}^+$ s.t.~there are two walks of length $l$ between nodes $v_i, v_j$ of different signs.
	\label{lem:unbalanced}
\end{lemma}
\begin{proof}
	We construct a directed signed graph $G'$ by making each edge in $G$ bidirectional in $G'$ while maintaining the same sign in both directions. We note that $G$ is strictly unbalanced if and only if $G'$ is strictly unbalanced, and also that the statement, i.e., $\exists v_i,v_j\in V$ and $l\in \mathbb{Z}^+$ s.t.~there are two walks of length $l$ between nodes $v_i, v_j$ of different signs, is true in $G$ if and only if it is true in $G'$. Hence, we prove the lemma through $G'$.
	
	By construction, $G'$ contains cycles of length $2$. (i) If $G'$ is periodic, then $G'$ is bipartite, because of the presence of length-$2$ cycle(s). Then all cycles have even length, and for each cycle $C$, we can find node $v_i, v_j\in C$, s.t.~the part starting from $v_i$ to $v_j$ has the same length as the remaining part from $v_j$ back to $v_i$. Since each edge is bidirectional, it means that we can find two walks of the same length from $v_i$ to $v_j$. Then, suppose the statement is not true, i.e., all walks of the same length between each pair of nodes $v_h,v_l\in V$ have the same sign, then all cycles are positive, noting that the two edges connecting the same pair of nodes have the same sign, thus $G'$ is balanced, which leads to contradiction. (ii) Otherwise, $G'$ is aperiodic, then the statement is true by Proposition 3.5 in \cite{Li_voterIM_2015}. 
\end{proof}
\begin{theorem}
	$G$ is strictly unbalanced if and only if $\rho(\mathbf{W}) < \rho(\bar{\mathbf{W}})$.
	\label{the:strict-unb-rho}
\end{theorem}
\begin{proof}
	We first note that if $\rho(\mathbf{W}) < \rho(\bar{\mathbf{W}})$, then $G$ is strictly unbalanced, since the spectral radius will be the same if $G$ is balanced or antibalanced by Theorem \ref{the:transition-spect}. 
	
	For the other direction, if $G$ is strictly unbalanced, by Lemma \ref{lem:unbalanced}, $\exists v_i,v_j\in V$ and $l_1\in \mathbb{Z}^+$ s.t.~there are two walks of length $l_1$ between nodes $v_i, v_j$ of different signs. Then
	\begin{align*}
		\abs{(\mathbf{W}^{l_1})_{ij}} <  (\bar{\mathbf{W}}^{l_1})_{ij}, 
	\end{align*}
	where $(\mathbf{W})_{ij}$ indicates the $(i,j)$ element of a matrix $\mathbf{W}$. Hence, for sufficiently large $l_2$, the walks between each pair of nodes will be able to go through the two walks of different signs between nodes $v_i, v_j$, thus $\forall v_h,v_k\in V$,
	\begin{align*}
		\abs{(\mathbf{W}^{l_2})_{hk}} < 
		(\bar{\mathbf{W}}^{l_2})_{hk}.  
	\end{align*}
	Then for each vector $\mathbf{x} = (x_h)\in \mathbb{R}^n$ and $\norm{\mathbf{x}}_2 = 1$, we can find $\bar{\mathbf{x}} = (\abs{x_h})$ s.t.~$\norm{\bar{\mathbf{x}}}_2 = 1$ and
	\begin{align*} 
		\abs{(\mathbf{W}^{l_2}\mathbf{x})_h} = \abs{\sum_{k}(\mathbf{W}^{l_2})_{hk}x_k} \le \sum_{k}\abs{(\mathbf{W}^{l_2})_{hk}x_k} < \sum_{k}(\bar{\mathbf{W}}^{l_2})_{hk}\abs{x_k} = (\bar{\mathbf{W}}^{l_2}\bar{\mathbf{x}})_h,
	\end{align*}
	where $(\mathbf{x})_h$ indicates the $h$-th element of an vector $\mathbf{x}$. Therefore, $\norm{\mathbf{W}^{l_2}\mathbf{x}}_2 < \norm{\bar{\mathbf{W}}^{l_2}\bar{\mathbf{x}}}_2$. Hence, by definition,
	\begin{align*}
		\norm{\mathbf{W}^{l_2}}_2 = \max_{\norm{\mathbf{x}}_2 = 1}\norm{\mathbf{W}^{l_2}\mathbf{x}}_2 < \max_{\norm{\mathbf{y}}_2 = 1}\norm{\bar{\mathbf{W}}^{l_2}\mathbf{y}}_2 = \norm{\bar{\mathbf{W}}^{l_2}}_2,
	\end{align*}
	then $\rho(\mathbf{W})^{l_2} = \rho(\mathbf{W}^{l_2}) < \rho(\bar{\mathbf{W}}^{l_2}) = \rho(\bar{\mathbf{W}})^{l_2}$, and finally $\rho(\mathbf{W}) < \rho(\bar{\mathbf{W}})$. 
\end{proof}

\subsection{Dynamical properties}
\label{sec:dynamicprop}
Characterized by different structural properties, the classification we have introduced also provide a way to characterize the dynamics happening on signed networks. Here, we consider two very different dynamics, where one is linear, the linear adjacency dynamics, and the other is nonlinear, the extended linear threshold (ELT) model. Furthermore, we consider random walks as a slightly modified example of the linear adjacency dynamics, and interpret the results accordingly.
% \begin{itemize}
	%     \item The linear adjacency dynamics: 
	%     \begin{itemize}
		%         \item The dynamics: how to extend to signed networks. 
		%         \item The evolution of the state values: balanced and antibalanced networks. 
		%         \item The measures for strictly unbalanced networks: example of transition matrix.
		%     \end{itemize}
	%     \item The extended linear threshold model: the evolution of the state values. 
	% \end{itemize}

\subsubsection{Linear adjacency dynamics}
\label{sec:linear_adj_dynamics}
In unsigned networks, linear adjacency dynamics sum over the state values of one's neighbours in the previous step to obtain the current state value of each node, since there are only positive edges. While in signed networks, we represent the signed weighted adjacency matrix as $\mathbf{W} = \mathbf{W}^+ - \mathbf{W}^-$, where $W_{ij}^+ = W_{ij}$ if $W_{ij} > 0$ ($0$ otherwise) and $W_{ij}^- = -W_{ij}$ if $W_{ij} < 0$ ($0$ otherwise). To start with, we consider two components of the state values, positive and negative parts, and we represent them on each node $v_j$ as $x_j^+$ and $x_j^-$, respectively. From the opposing rule, we assume that negative edges can flip the sign of the corresponding state values, while positive ones will remain the sign. Hence, there are two different possibilities for the positive part on $v_j$, either from positive parts through positive edges or from negative parts through negative edges, i.e.,
\begin{align}
	x_j^+(t) = \sum_{i}W_{ij}^+x^+_i(t-1) + W_{ij}^-x^-_i(t-1);
	\label{equ:ld+}
\end{align}
there are also two possibilities for the negative part on $v_j$, either from positive parts through negative edges or from negative parts through positive edges, i.e.,
\begin{align}
	x_j^-(t) = \sum_{i}W_{ij}^-x^+_i(t-1) + W_{ij}^+x^-_i(t-1).
	\label{equ:ld-}
\end{align}
There could be different approaches to combine these two parts. If we concatenate the negative one to the positive one directly, the whole dynamics is thus governed by a $2n\times 2n$ matrix $\mathbf{W}^{(2)} = [\mathbf{W}^+, \mathbf{W}^-; \mathbf{W}^-, \mathbf{W}^+]$, similar to \eqref{equ:A2}, while if we simply sum the two parts, it will recover the dynamics on the unsigned counterpart $\bar{G}$. Here specifically, we consider the ``polarization" on each node by subtracting the negative part in \eqref{equ:ld-} from the positive part in \eqref{equ:ld+} as the state value, where $\forall v_j\in V,\ t > 0$,
\begin{align}
	x_j(t) = x_j^+(t) - x_j^-(t) = \sum_{i}\left(W_{ij}^+x_i(t-1) - W_{ij}^-x_i(t-1)\right) = \sum_{i}W_{ij}x_i(t-1),
	\label{equ:signed_linear_dynamics}
\end{align}
where the initial vector $\mathbf{x}(0)$ is given. Hence, the state vector at each time step $t \ge 0$ is 
\begin{align*}
	\mathbf{x}(t)^T = \mathbf{x}(0)^T\mathbf{W}^t,
\end{align*} 
and the evolution of $\mathbf{W}^t$ over time provides insights into the behavior of the linear adjacency dynamics. 
% Note that we consider the time-discounted sum of state values
% \begin{align*}
	%     s_j = \sum_{t=1}^\infty (1-\gamma)^t x_j(t),
	% \end{align*}
% as an important quantity for each node $v_j$, e.g., encoding the overall influence on the node, where the factor $\gamma$ satisfies $\gamma > 1-1/\rho(\mathbf{W})$,
% % \begin{equation}
	% %     \gamma > 1-1/\rho(\mathbf{W}),
	% %     \label{equ:signed_linear-dyn_req}
	% % \end{equation}
% with $\rho(\mathbf{W})$ being the spectral radius. Hence, the time-discounted state values approach $0$ as $t$ goes to infinity, $\lim_{t\to\infty}(1-\gamma)^t x_j(t) = 0$
% % \begin{align*}
	% %     \lim_{t\to\infty}(1-\gamma)^t x_j(t) = 0. 
	% % \end{align*}

Starting from the unitary decomposition, $\mathbf{W} = \mathbf{U}\Lambda\mathbf{U}^{T} = \sum_{i=1}^{n}\lambda_i\mathbf{u}_i\mathbf{u}_i^T$, where $\lambda_1\ge \lambda_2 \ge \dots \ge \lambda_n$ are the eigenvalues of $\mathbf{W}$ and $\mathbf{u}_1, \mathbf{u}_2, \dots, \mathbf{u}_n$ are the associated eigenvectors, we have 
\begin{align}
	\mathbf{W}^t = \sum_{i=1}^{n}\lambda_i^t\mathbf{u}_i\mathbf{u}_i^T.
\end{align}
Hence, the behavior of $\mathbf{W}^t$ is dominated by the eigenvectors associated with the eigenvalues significantly different from $0$, and for sufficiently large $t$, the behavior could be approximated by that of the leading eigenvalue and the associated eigenvector(s). In the following, we consider different balanced structures in signed networks, specifically balance, antibalance and strict unbalance. We denote the eigenvalues of the weighted adjacency matrix ignoring the edge sign $\bar{\mathbf{W}}$ by $\bar{\lambda}_1\ge \bar{\lambda}_2 \ge \dots \ge \bar{\lambda}_n$, with the associated eigenvectors $\bar{\mathbf{u}}_1, \bar{\mathbf{u}}_2, \dots, \bar{\mathbf{u}}_n$, where $\bar{\mathbf{W}} = \sum_{i=1}^n\bar{\lambda}_i\bar{\mathbf{u}}_i\bar{\mathbf{u}}_i^T$. We denote the bipartition corresponding to the balanced or antibalanced structure by $V_1, V_2$, and the diagonal matrix corresponding to the partition by $\mathbf{S}$ with its $(i,i)$ element being $1$ if $v_i\in V_1$ and $-1$ otherwise.

\paragraph{Balanced networks.} If the network is balanced with bipartition $V_1, V_2$, then from Theorem \ref{the:transition-spect}, we have
\begin{align}
	\mathbf{W}^t = \mathbf{S}\left(\sum_{i=1}^{n}\bar{\lambda}_i^t\bar{\mathbf{u}}_i\bar{\mathbf{u}}_i^T\right)\mathbf{S} = \mathbf{S}\bar{\mathbf{W}}^t\mathbf{S},
	\label{equ:signed-linear_b-Wt}
\end{align}
thus the magnitude of the elements in $\mathbf{W}^t$ evolves as in the simple network ignoring the edge sign (i.e.~$\bar{\mathbf{W}}^t$), and the difference lies in the sign pattern (of nonzero elements): the $(i,j)$ element is positive if nodes $v_i,v_j\in V_1$ or $v_i,v_j\in V_2$ and negative otherwise. This means that at each time step $t$, nodes tend to have the same sign as the ones in the same node subset of the bipartition, while the opposite sign from the others. Further, if the network is irreducible and aperiodic, then when $t$ is sufficiently large, $\mathbf{W}^t$ can be well approximated by its rank-1 approximation,
\begin{align*}
	\hat{\mathbf{W}}^t = \bar{\lambda}_1^t\mathbf{S}\bar{\mathbf{u}}_1\bar{\mathbf{u}}_1^T\mathbf{S},
\end{align*}
where $\norm{\mathbf{W}^t - \hat{\mathbf{W}}^t}_F = \sqrt{\sum_{i=2}^{n}\bar{\lambda}_i^{2t}}$, with $\norm{\cdot}_F$ denoting the Frobenius norm, by noting that $\lim_{t\to\infty}\abs{\bar{\lambda}_i^{t}/\bar{\lambda}_1^t} = 0,\, \forall i\ne 1$, from Proposition \ref{pro:transition-spect-rho}. Hence, the asymptotic behavior can be approximated by the term associated with the leading eigenvalue and the associated eigenvector, which has the same sign pattern as \eqref{equ:signed-linear_b-Wt}.

\paragraph{Antibalanced networks.} If the network is antibalanced with bipartition $V_1, V_2$, then from Theorem \ref{the:transition-spect}, we have 
\begin{align}
	\mathbf{W}^t = \mathbf{S}\left(\sum_{i=1}^{n}(-\bar{\lambda}_i)^t\bar{\mathbf{u}}_i\bar{\mathbf{u}}_i^T\right)\mathbf{S} = (-1)^t\mathbf{S}\bar{\mathbf{W}}^t\mathbf{S},
	\label{equ:signed-linear_antib-Wt}
\end{align}
thus again, the magnitude of the elements in $\mathbf{W}^t$ evolves as in the simple network ignoring the edge sign, but here the sign pattern (of nonzero elements) alternates over time: when $t$ is odd, the $(i,j)$ element is \textit{negative} if nodes $v_i,v_j\in V_1$ or $v_i,v_j\in V_2$ and positive otherwise; while $t$ is even in the following step, the $(i,j)$ element becomes \textit{positive} if nodes $v_i,v_j\in V_1$ or $v_i,v_j\in V_2$ and negative otherwise. Hence, the antibalanced structure is highly unstable. 
% Further, if the network is irreducible and aperiodic, then similar to the case of balanced networks, when $t$ is sufficiently large, $\mathbf{W}^t$ can be well approximated by its rank-1 approximation,
% \begin{align*}
	%     \hat{\mathbf{W}}^t = (-\bar{\lambda}_1)^t\mathbf{S}\bar{\mathbf{u}}_1\bar{\mathbf{u}}_1^T\mathbf{S},
	% \end{align*}
% which has the same sign pattern as \eqref{equ:signed-linear_antib-Wt}.

\paragraph{Strictly unbalanced networks.} In all the remaining networks, neither are they like balanced networks where all walks of length $t+1$ have the same sign as the walks of length $t$ connecting each pair of nodes $v_i,v_j$, nor are they like antibalanced networks where all walks of length $t+1$ have the opposite sign as the walks of length $t$ connecting each pair of nodes $v_i,v_j$, for each $t > 0$. Hence to characterize its performance, we propose the following measures to quantify how far a network is from being balanced or antibalanced, motivated by the signed Cheeger inequality \cite{Atay_signedCheeger_2020}: for the distance from being balanced, we have
\begin{align}
	d_b(G) = \lambda_{min}(\mathbf{L}_{rw}(G)),
	\label{equ:signed_d-balanced}
\end{align} 
and for the distance from being antibalanced, we have
\begin{align}
	d_a(G) = 2 - \lambda_{max}(\mathbf{L}_{rw}(G)),
	\label{equ:signed_d-anti-balanced}
\end{align}
where $\mathbf{L}_{rw}(G)$ is the random walk Laplacian of the signed network $G$ as in \eqref{equ:signed_Lrw}, and $\lambda_{min}(\cdot), \lambda_{max}(\cdot)$ return the smallest and the largest eigenvalues, respectively. We note that in unsigned networks, the smallest eigenvalue of the random walk Laplacian is trivially $0$ (corresponding to that the transition matrix has a trivial largest eigenvalue $1$) and the smallest nontrivial one is important from many aspects, including the relaxation time of random walks \cite{Masuda2017RW}. However, in signed networks, the smallest eigenvalue is nontrivial. We also note that being close to balance does not necessarily indicate the signed network is far from antibalance, where, for example, for bipartite networks, balance indicates antibalance. Hence, we maintain two measures for the two dimensions. We will further interpret both measures in subsection \ref{sec:random_walk}. 

Therefore, depending on how far the signed network is from being balanced or antibalanced, it can have performance closer to that of balanced or antibalanced networks.
\begin{enumerate}[label=(\roman*)]
	\item If $d_b(G) < d_a(G)$, we expect $G$ to be closer to being balanced. Then $\forall v_i,v_j\in V, t > 0$, we expect that most walks of length $t+1$ connecting nodes $v_i,v_j$ have the same sign as most walks of length $t$ (if any), thus $\mathbf{W}^t$ tends to maintain the same sign pattern over time.
	\item If $d_b(G) > d_a(G)$, we expect $G$ to be closer to being antibalanced. Then $\forall v_i,v_j\in V, t > 0$, we expect that most walks of length $t+1$ connecting nodes $v_i,v_j$ have the opposite sign as most walks of length $t$ (if any), thus $\mathbf{W}^t$ tends to alternate the sign pattern over time. 
\end{enumerate}
From Theorem \ref{the:strict-unb-rho}, $\rho(\mathbf{W}) < \rho(\bar{\mathbf{W}})$ where $\rho(\cdot)$ is the spectral radius or the eigenvalue of the largest magnitude, thus when $t$ is sufficiently large, elements in $\mathbf{W}^t$ will have smaller magnitude than those in $\bar{\mathbf{W}}^t$.

\subsubsection{An example: random walks}
\label{sec:random_walk}
Here, we consider random walks as a specific example of the linear adjacency dynamics, with some modifications. Note that the adjacency matrix of an undirected network has to be symmetric, but it is not necessarily the case for the transition matrix $\mathbf{P} = \mathbf{D}^{-1}\mathbf{W}$. However, it is similar to a symmetric matrix $\mathbf{P}_{sym} = \mathbf{D}^{-1/2}\mathbf{W}\mathbf{D}^{-1/2}$, where $\mathbf{P} = \mathbf{D}^{-1/2}\mathbf{P}_{sym}\mathbf{D}^{1/2}$, and for each eigenpair $(\lambda, \mathbf{D}^{1/2}\mathbf{x})$ of $\mathbf{P}_{sym}$, $(\lambda, \mathbf{x})$ is also an eigenpair of $\mathbf{P}$. Hence, the above results in subsection \ref{sec:linear_adj_dynamics} can still be applied, but indirectly through $\mathbf{P}_{sym}$. We denote the eigenvalues as $\lambda_1\ge \dots \ge \lambda_n$ with the associated (right) eigenvectors of $\mathbf{P}$ as $\mathbf{u}_1, \dots, \mathbf{u}_n$, thus the eigenvectors of $\mathbf{P}_{sym}$ are $\mathbf{D}^{1/2}\mathbf{u}_1, \dots, \mathbf{D}^{1/2}\mathbf{u}_n$. For illustrative purposes, we only assume $\mathbf{D}^{1/2}\mathbf{u}_1, \dots, \mathbf{D}^{1/2}\mathbf{u}_n$ to be orthonormal in this section. We denote the unsigned counterparts as $\bar{\mathbf{P}}$ and $\bar{\mathbf{P}}_{sys}$, and their eigenvalues as $\bar{\lambda}_1\ge \dots \ge \bar{\lambda}_n$.

\paragraph{Balanced networks.} We start with the case when the signed network is structurally balanced, and will show that a steady state is achievable. From Theorem \ref{the:transition-spect}, we can easily deduct the characteristics of the leading eigenvalue as in Corollary \ref{cor:balance-lambda} (see Appendix \ref{sec:app_proofs} for a detailed proof). 
\begin{corollary}
	The signed transition matrix $\mathbf{P}$ has eigenvalue $1$ if and only if $G$ is balanced. 
	\label{cor:balance-lambda}
\end{corollary}
We also note that $\mathbf{L}_{rw} = \mathbf{I} - \mathbf{P}$, thus $1$ being an eigenvalue of $\mathbf{P}$ is equivalent to $0$ being an eigenvalue of $\mathbf{L}_{rw}$, where the latter has been shown as a sufficient and necessary condition for the graph being balanced \cite{hou_2003_Laplacian,Li_2009_normL,zaslavsky_1982_signed}. 

\begin{proposition}
	If $G$ is balanced, then $\mathbf{P}^t$ is still a signed transition matrix, and has the following signed pattern: 
	\begin{displaymath}
		(\mathbf{P}^t)_{ij} = 
		\begin{cases}
			(\bar{\mathbf{P}}^t)_{ij}, \quad &\text{if } v_i,v_j\in V_1 \text{ or } v_i,v_j\in V_2\\
			- (\bar{\mathbf{P}}^t)_{ij},\quad &\text{otherwise},
		\end{cases}
	\end{displaymath}
	where $V_1, V_2$ denote the bipartition corresponding to the balanced structure.
\end{proposition}
\begin{proof}
	If $G$ is balanced, then $\mathbf{P} = \mathbf{S}\bar{\mathbf{P}}\mathbf{S}$, where $\mathbf{S}$ is the diagonal matrix whose $(i,i)$ element is $1$ if $v_i\in V_1$ and $-1$ otherwise. Then 
	\begin{align*}
		\mathbf{P}^t = \left(\mathbf{S}\bar{\mathbf{P}}\mathbf{S}\right)^t = \mathbf{S}\bar{\mathbf{P}}^t\mathbf{S}.
	\end{align*}
	Since $\bar{\mathbf{P}}^t$ is still a transition matrix, $\mathbf{P}^t$ is a signed transition matrix. Meanwhile, $(\mathbf{P}^t)_{ij} = (\bar{\mathbf{P}}^t)_{ij}(\mathbf{S})_{ii}(\mathbf{S})_{jj}$, hence is $(\bar{\mathbf{P}}^t)_{ij}$ if $v_i,v_j\in V_1$ or $v_i,v_j\in V_2$, and $-(\bar{\mathbf{P}}^t)_{ij}$ otherwise. 
\end{proof}

\begin{proposition}
	If $G$ is balanced and is not bipartite, then the steady state is $\mathbf{x}^* = (x_j^*)$ where
	\begin{displaymath}
		x_j^* = 
		\begin{cases}
			(\mathbf{x}(0)^T\mathbf{1}_1) d_j/(2m),\quad &\text{if } v_j\in V_1,\\
			- (\mathbf{x}(0)^T\mathbf{1}_1) d_j/(2m),\quad &\text{otherwise,}
		\end{cases}
	\end{displaymath}
	where $\mathbf{x}(0) = (x_i(0))$ is the initial state vector with $\sum_i\abs{x_i(0)} = 1$, $\mathbf{1}_1$ is the diagonal vector of $\mathbf{S}$ with the $i$-th element being $1$ if $v_i\in V_1$ and $-1$ otherwise, and $2m = \sum_{j}d_j$. 
	\label{pro:balance-steady}
\end{proposition}
\begin{proof}
	Since $G$ is not bipartite, $\abs{\lambda_i} < 1,\ \forall i\ne 1$. Hence, 
	\begin{align*}
		\lim_{t\to\infty}\mathbf{P}_{sys}^t = \lim_{t\to\infty}\sum_{i=1}^n\lambda_i^t\left(\mathbf{D}^{1/2}\mathbf{u}_i\right)\left(\mathbf{D}^{1/2}\mathbf{u}_i\right)^T = \left(\mathbf{D}^{1/2}\mathbf{u}_1\right)\left(\mathbf{D}^{1/2}\mathbf{u}_1\right)^T,
	\end{align*}
	where the eigenvectors $\mathbf{D}^{1/2}\mathbf{u}_i$ are orthonormal to each other, and $\mathbf{u}_i$ is the eigenvector of $\mathbf{P}$ associated with the same eigenvalue. By Proposition \ref{pro:transition-spect-rho}, $\mathbf{u}_1$ has the specific structure where $\mathbf{u}_1 = c\mathbf{1}_1$ for some nonzero constant $c\in \mathbb{R}$. WOLG, we assume $c > 0$. Then since $\mathbf{D}^{1/2}\mathbf{u}_i$ has $2$-norm $1$, $c = 1/\sqrt{2m}$. Hence, 
	\begin{align*}
		\mathbf{x}^* 
		&= \lim_{t\to\infty}\mathbf{x}(0)^T\mathbf{P}^t\\
		&= \lim_{t\to\infty}\mathbf{x}(0)^T\mathbf{D}^{-1/2}\mathbf{P}_{sys}^t\mathbf{D}^{1/2}\\
		&= \mathbf{x}(0)^T\mathbf{D}^{-1/2}\left(\mathbf{D}^{1/2}\mathbf{u}_1\right)\left(\mathbf{D}^{1/2}\mathbf{u}_1\right)^T\mathbf{D}^{1/2}\\
		&= \mathbf{x}(0)^T(c\mathbf{1}_1)(c\mathbf{1}_1)^T\mathbf{D} = \mathbf{x}(0)^T\mathbf{1}_1/(2m)\mathbf{1}_1^T\mathbf{D}. 
	\end{align*}
	Hence, $x^*_j = (\mathbf{x}(0)^T\mathbf{1}_1)d_j/(2m)$ if $v_j\in V_1$ and $-(\mathbf{x}(0)^T\mathbf{1}_1)d_j/(2m)$ otherwise.
\end{proof}
Hence, from Proposition \ref{pro:balance-steady}, we can see that the steady state now depends on the initial condition, which is different from random walks defined on networks only of positive connections. However, if we further assume that the initial state agrees with the balanced structure, where it has positive values in one node subset of the bipartition (e.g., $V_1$) and negative values in the other (e.g., $V_2$), the dependence can be removed partially since $\abs{\mathbf{x}(0)^T\mathbf{1}_1} = 1$, while the sign of the steady state still depends on the initialization.

\paragraph{Antibalanced networks.} We then continue to the case when the signed network is antibalanced, and will show that a steady state cannot be achieved generally, where there is one limit for odd times and another for even times. To start with, the leading eigenvalue can be characterized as in Corollary \ref{cor:antibalance-lambda}, which can be directly deducted from Theorem \ref{the:transition-spect} (or Corollary \ref{cor:balance-lambda}).
\begin{corollary}
	The signed transition matrix $\mathbf{P}$ has eigenvalue $-1$ if and only if $G$ is antibalanced.
	\label{cor:antibalance-lambda}
\end{corollary} 
% \begin{proof}
	%     A graph $G = (V, E, \mathbf{W})$ is antibalanced if and only if the graph constructed by negating the edge sign $G_n = (V, E, -\mathbf{W})$ is balanced. By Corollary \ref{cor:balance-lambda}, $G_n$ is balanced if and only if its signed transition matrix $\mathbf{P}_n = - \mathbf{P}$ has eigenvalue $1$, which is equivalent to that $\mathbf{P}$ has eigenvalue $-1$. 
	% \end{proof}
Note again that $-1$ being an eigenvalue of $\mathbf{P}$ is equivalent to $2$ being an eigenvalue of $\mathbf{L}_{rw}$, where the latter has also been shown for antibalanced graphs \cite{Li_2009_normL}.

\begin{proposition}
	If $G$ is antibalanced, then $\mathbf{P}^t$ is still a signed transition matrix, and has the following signed pattern: 
	\begin{displaymath}
		(\mathbf{P}^t)_{ij} = 
		\begin{cases}
			(-1)^t(\bar{\mathbf{P}}^t)_{ij}, \quad &\text{if } v_i,v_j\in V_1 \text{ or } v_i,v_j\in V_2\\
			(-1)^{t+1} (\bar{\mathbf{P}}^t)_{ij},\quad &\text{otherwise},
		\end{cases}
	\end{displaymath}
	where $V_1,V_2$ denote the bipartition corresponding to the antibalanced structure.
\end{proposition}
\begin{proof}
	If $G$ is antibalanced, then $\mathbf{P} = -\mathbf{S}\bar{\mathbf{P}}\mathbf{S}$, where $\mathbf{S}$ is the diagonal matrix whose $(i,i)$ element is $1$ if $v_i\in V_1$ and $-1$ otherwise. Then
	\begin{align*}
		\mathbf{P}^t = \left(-\mathbf{S}\bar{\mathbf{P}}\mathbf{S}\right)^t = (-1)^t\mathbf{S}\bar{\mathbf{P}}^t\mathbf{S}.
	\end{align*}
	Since $\bar{\mathbf{P}}^t$ is still a transition matrix, $\mathbf{P}^t$ is a signed transition matrix. Meanwhile, $(\mathbf{P}^t)_{ij} = (-1)^t(\bar{\mathbf{P}}^t)_{ij}(\mathbf{S})_{ii}(\mathbf{S})_{jj}$, hence is $(-1)^t(\bar{\mathbf{P}}^t)_{ij}$ if $v_i,v_j\in V_1$ or $v_i,v_j\in V_2$, and $(-1)^{t+1}(\bar{\mathbf{P}}^t)_{ij}$ otherwise.
\end{proof}

\begin{proposition}
	If $G$ is antibalanced and is not bipartite, then the random walks do not have a steady state, where there are two different limits for odd or even times, denoted by $\mathbf{x}^{*o} = (x_j^{*o})$ and $\mathbf{x}^{*e} = (x_j^{*e})$, respectively, where
	\begin{displaymath}
		x_j^{*o} = 
		\begin{cases}
			-(\mathbf{x}(0)^T\mathbf{1}_1) d_j/(2m),\quad &\text{if } v_j\in V_1,\\
			(\mathbf{x}(0)^T\mathbf{1}_1) d_j/(2m),\quad &\text{otherwise,}
		\end{cases}
	\end{displaymath}
	while 
	\begin{displaymath}
		x_j^{*e} = 
		\begin{cases}
			(\mathbf{x}(0)^T\mathbf{1}_1) d_j/(2m),\quad &\text{if } v_j\in V_1,\\
			-(\mathbf{x}(0)^T\mathbf{1}_1) d_j/(2m),\quad &\text{otherwise,}
		\end{cases}
	\end{displaymath}
	where $\mathbf{x}(0)$ is the initial state, $\mathbf{1}_1$ is the diagonal vector of $\mathbf{S}$ with the $i$-th element being $1$ if $v_i\in V_1$ and $-1$ otherwise, and $2m = \sum_{j}d_j$. 
	\label{pro:antibalance-steady}
\end{proposition}
\begin{proof}
	Since $G$ is not bipartite, $\abs{\lambda_i} < 1,\, \forall i\ne n$. Hence, 
	for odd times, 
	\begin{align*}
		\lim_{t\to\infty}\mathbf{P}_{sys}^{2t-1} 
		&= \lim_{t\to\infty}\sum_{i=1}^n\lambda_i^{2t-1}\left(\mathbf{D}^{1/2}\mathbf{u}_i\right)\left(\mathbf{D}^{1/2}\mathbf{u}_i\right)^T\\
		&= \lim_{t\to\infty}(-1)^{2t-1}\left(\mathbf{D}^{1/2}\mathbf{u}_n\right)\left(\mathbf{D}^{1/2}\mathbf{u}_n\right)^T\\
		&= -\left(\mathbf{D}^{1/2}\mathbf{u}_n\right)\left(\mathbf{D}^{1/2}\mathbf{u}_n\right)^T,
	\end{align*}
	and similarly for even times, 
	\begin{align*}
		\lim_{t\to\infty}\mathbf{P}_{sys}^{2t} = \left(\mathbf{D}^{1/2}\mathbf{u}_n\right)\left(\mathbf{D}^{1/2}\mathbf{u}_n\right)^T,
	\end{align*}
	where the eigenvectors $\mathbf{D}^{1/2}\mathbf{u}_i$ are orthonormal to each other, and $\mathbf{u}_i$ is the eigenvector of $\mathbf{P}$ associated with the same eigenvalue. From Proposition \ref{pro:transition-spect-rho}, $\mathbf{u}_n$ has the specific structure where $\mathbf{u}_n = c\mathbf{1}_1$ for some nonzero constant $c\in \mathbb{R}$, and WLOG, $c > 0$. Then from $\mathbf{D}^{1/2}\mathbf{u}_i$ has $2$-norm $1$, $c = 1/\sqrt{2m}$. Hence, for odd times, 
	\begin{align*}
		\mathbf{x}^{*o} 
		&= \lim_{t\to\infty}\mathbf{x}(0)^T\mathbf{P}^{2t-1}\\
		&= \lim_{t\to\infty}\mathbf{x}(0)^T\mathbf{D}^{-1/2}\mathbf{P}_{sys}^{2t-1}\mathbf{D}^{1/2}\\
		&= -\mathbf{x}(0)^T\mathbf{D}^{-1/2}\left(\mathbf{D}^{1/2}\mathbf{u}_n\right)\left(\mathbf{D}^{1/2}\mathbf{u}_n\right)^T\mathbf{D}^{1/2}\\
		&= -\mathbf{x}(0)^T(c\mathbf{1}_1)(c\mathbf{1}_1)^T\mathbf{D} = -\mathbf{x}(0)^T\mathbf{1}_1/(2m)\mathbf{1}_1^T\mathbf{D},
	\end{align*}
	and similarly for even times, 
	\begin{align*}
		\mathbf{x}^{*2} 
		&= \lim_{t\to\infty}\mathbf{x}(0)^T\mathbf{P}^{2t}\\
		&= \lim_{t\to\infty}\mathbf{x}(0)^T\mathbf{D}^{-1/2}\mathbf{P}_{sys}^{2t}\mathbf{D}^{1/2}\\
		&= \mathbf{x}(0)^T\mathbf{D}^{-1/2}\left(\mathbf{D}^{1/2}\mathbf{u}_n\right)\left(\mathbf{D}^{1/2}\mathbf{u}_n\right)^T\mathbf{D}^{1/2}\\
		&= \mathbf{x}(0)^T\mathbf{1}_1/(2m)\mathbf{1}_1^T\mathbf{D},
	\end{align*}
	which are of the same forms as stated.
\end{proof}
Hence, from Proposition \ref{pro:antibalance-steady}, we can see that in antibalanced networks, the limiting behavior of signed random walks not only depends on the initial condition, but also odd or even times. 

\paragraph{Strictly unbalanced networks.} Finally, we consider all the remaining signed networks, the strictly unbalanced ones. Interestingly, a steady state is actually achievable in this case.
\begin{proposition}
	If $G$ is strictly unbalanced, then the steady state is $\mathbf{x}^*=\mathbf{0}$, where $\mathbf{0}$ is the vector of zeros.
\end{proposition}
\begin{proof}
	When $G$ is strictly unbalanced, by Theorem \ref{the:strict-unb-rho}, $\rho(\mathbf{P}_{sys}) < \rho(\bar{\mathbf{P}}_{sys}) = 1$. Hence,
	\begin{align*}
		\lim_{t\to\infty}\mathbf{P}_{sys}^t = \lim_{t\to\infty}\sum_{t=1}^n\lambda_i^t\left(\mathbf{D}^{1/2}\mathbf{u}_i\right)\left(\mathbf{D}^{1/2}\mathbf{u}_i\right)^T = \mathbf{O},
	\end{align*}
	where $\mathbf{O}$ is the matrix of zeros. Hence, 
	\begin{align*}
		\mathbf{x}^* = \lim_{t\to\infty}\mathbf{x}(0)^T\mathbf{P}^t = \lim_{t\to\infty}\mathbf{x}(0)^T\mathbf{D}^{-1/2}\mathbf{P}_{sys}^t\mathbf{D}^{1/2} = \mathbf{x}(0)^T\mathbf{D}^{-1/2}\mathbf{O}\mathbf{D}^{1/2} = \mathbf{0}.
	\end{align*}
\end{proof}

We know that when $G$ is balanced, $d_b(G) = 0$, and when $G$ is antibalanced, $d_a(G) = 0$, but the problem remains what these measures correspond to quantitatively in the case of strictly unbalanced networks. We note that for each eigenvalue $\lambda$ of $\mathbf{L}_{rw}$, $1-\lambda$ is also an eigenvalue of $\mathbf{P}$. Hence, $d_b(G) = 1 - \lambda_{\max}(\mathbf{P}(G))$ and $d_a(G) = 1 + \lambda_{\min}(\mathbf{P}(G))$. Then if we denote the balanced network that is closest to $G$ by $G^b$, while the antibalanced network that is closest to $G$ by $G^a$, in the sense that $G$ will become balanced or antibalanced by flipping the sign of the least number of edges, then measures \eqref{equ:signed_d-balanced} and \eqref{equ:signed_d-anti-balanced} are 
\begin{align*}
	d_b(G) &= - \left(\lambda_{\max}(\mathbf{P}(G)) - 1\right) = - \left(\lambda_{\max}(\mathbf{P}(G)) - \lambda_{\max}(\mathbf{P}(G^b))\right) = -\Delta_{\max},\\
	d_a(G) &= \lambda_{\min}(\mathbf{P}(G)) - (-1) = \lambda_{\min}(\mathbf{P}(G)) - \lambda_{\min}(\mathbf{P}(G^a)) = \Delta_{\min},    
\end{align*}
where we denote $\lambda_{\max}(\mathbf{P}(G)) - \lambda_{\max}(\mathbf{P}(G^b))$ by $\Delta_{\max}$, and denote $\lambda_{\min}(\mathbf{P}(G)) - \lambda_{\min}(\mathbf{P}(G^a))$ by $\Delta_{\min}$. Hence, we analyze the measures through considering the strictly unbalanced network $G$ as a perturbation of a balanced or antibalanced network, whichever is closer to $G$. In the following, we show exclusively the results from perturbing balanced networks, and the results for the antibalanced ones follow similarly. We also refer the reader to \cite{Greenbaum_pertub_2020} and references therein for more mathematical details.

Let us denote the signed transition matrix of $G^b$ by $\mathbf{P}^b$, and its largest eigenvalues by $\lambda^b_1 = 1$ where $\mathbf{P}^b\mathbf{u}^b = \lambda^b_1\mathbf{u}^b$ and $\mathbf{w}^{bT}\mathbf{P}^{b} = \lambda^b_1\mathbf{w}^{bT}$ with $\mathbf{u}^b,\mathbf{w}^b$ denoting the right and left eigenvectors, respectively, of $\mathbf{P}^b$. We consider $\mathbf{P} = \mathbf{P}^b + \Delta\mathbf{P}$. If $\Delta\mathbf{P}$ is small and the largest eigenvalue of $\mathbf{P}^b$ is well separated from the others, we can consider the largest eigenvalue of $\mathbf{P}$ as perturbing the one of $\mathbf{P}^b$, where $\lambda_1 = \lambda^b_1 + \Delta_{\max}$ with its corresponding right eigenvector $\mathbf{u} = \mathbf{u}^b + \Delta\mathbf{u}$. Hence, 
\begin{align*}
	(\mathbf{P}^b + \Delta\mathbf{P})(\mathbf{u}^b + \Delta\mathbf{u}) = (\lambda^b_1 + \Delta_{\max})(\mathbf{u}^b + \Delta\mathbf{u}). 
\end{align*}
Then, left multiplying $\mathbf{w}^{bT}$ gives, 
\begin{align*}
	\mathbf{w}^{bT}(\mathbf{P}^b + \Delta\mathbf{P})(\mathbf{u}^b + \Delta\mathbf{u}) 
	&= (\lambda^b_1 + \Delta_{\max})\mathbf{w}^{bT}(\mathbf{u}^b + \Delta\mathbf{u})\\
	%\mathbf{y}^{bT}\mathbf{P}^b\mathbf{x}^b + \mathbf{y}^{bT}\mathbf{P}^b\Delta\mathbf{x} + \mathbf{y}^{bT}\Delta\mathbf{P}\mathbf{x}^b + \mathbf{y}^{bT}\Delta\mathbf{P}\Delta\mathbf{x} &= \lambda^b_1\mathbf{y}^{bT}\mathbf{x}^b + \lambda^b_1\mathbf{y}^{bT}\Delta\mathbf{x} + \Delta_{\max}\mathbf{y}^{bT}\mathbf{x}^b + \Delta_{\max}\mathbf{y}^{bT}\Delta\mathbf{x}\\ 
	\mathbf{w}^{bT}\Delta\mathbf{P}\mathbf{u}^b + \mathbf{w}^{bT}\Delta\mathbf{P}\Delta\mathbf{u} &= \Delta_{\max}\mathbf{w}^{bT}\mathbf{u}^b + 
	\Delta_{\max}\mathbf{w}^{bT}\Delta\mathbf{u}.
\end{align*}
Since we assume $\Delta\mathbf{P}$ is small, we ignore second-order terms $\mathbf{w}^{bT}\Delta\mathbf{P}\Delta\mathbf{u}$ and $\Delta_{\max}\mathbf{w}^{bT}\Delta\mathbf{u}$, and then
\begin{align}
	\Delta_{\max} = \frac{\mathbf{w}^{bT}\Delta\mathbf{P}\mathbf{u}^b}{\mathbf{w}^{bT}\mathbf{u}^b}. 
	\label{equ:delta-lambda}
\end{align}

From Theorem \ref{the:transition-spect}, we can directly specify the left and right eigenvectors as in Corollary \ref{cor:balance-lreig}. We then interpret the proposed measures in Proposition \ref{pro:P-perturb}.
\begin{corollary}
	For a balanced network $G^b$, the signed transition matrix $\mathbf{P}^b$ has a right eigenvector $\mathbf{u}^b$ and a left eigenvector $\mathbf{w}^b$ associated with the largest eigenvalue $\lambda^b_1 = 1$, where
	\begin{align}
		u^b_i = 
		\begin{cases}
			1,\quad &\text{if } v_i \in V_1,\\
			-1, \quad &\text{if } v_i\in V_2,
		\end{cases}
		\quad 
		w^b_i = 
		\begin{cases}
			d_i,\quad &\text{if } v_i \in V_1,\\
			-d_i, \quad &\text{if } v_i\in V_2.
		\end{cases}
	\end{align}
	\label{cor:balance-lreig}
\end{corollary}
% \begin{proof}
	%     We can check that 
	%     \begin{align*}
		%         \left(\mathbf{P}^b\mathbf{u}^b\right)_i = \sum_{j}P^b_{ij}u^b_j = \sum_{v_j\in V_1}\frac{W_{ij}}{d_i} - \sum_{v_j\in V_2}\frac{W_{ij}}{d_i} = 
		%         \begin{cases}
			%             1,\quad &\text{if } i \in V_1\\
			%             -1, \quad &\text{if } i\in V_2
			%         \end{cases}
		%         = u^b_i,
		%     \end{align*}
	%     hence $\mathbf{u}^b$ is a right eigenvector of $\mathbf{P}^b$. Similarly,
	%     \begin{align*}
		%         \left(\mathbf{w}^{bT}\mathbf{P}^b\right)_i = \sum_{j}w^b_jP^b_{ji} = \sum_{v_j\in V_1}d_j\frac{W_{ji}}{d_j} - \sum_{v_j\in V_2}d_j\frac{W_{ji}}{d_j} = \begin{cases}
			%         d_i,\quad &\text{if } v_i \in V_1\\
			%         -d_i, \quad &\text{if } v_i\in V_2
			%     \end{cases}
		%     = w^b_i,
		%     \end{align*}
	%     hence $\mathbf{w}^b$ is a left eigenvector of $\mathbf{P}^b$.
	% \end{proof}

\begin{proposition}
	When $G^b$ can be obtained by flipping the sign of a set of edges $\tilde{E} \subset E$,
	\begin{align}
		\frac{\mathbf{w}^{bT}\Delta\mathbf{P}\mathbf{u}^b}{\mathbf{w}^{bT}\mathbf{u}^b} = -\frac{2\sum_{(v_i,v_j)\in \tilde{E}}\abs{W_{ij}}}{m}, 
	\end{align}
	where $2m = \sum_{j}d_j$.
	\label{pro:P-perturb}
\end{proposition}
\begin{proof}
	% When $G^b$ can be obtained by flipping the sign of one edge $(v_i,v_j)\in E$, 
	% \begin{align*}
		%     \Delta\mathbf{P} = -2P^b_{ij}\mathbf{e}_i\mathbf{e}_j^T - 2P^b_{ji}\mathbf{e}_j\mathbf{e}_i^T,
		% \end{align*}
	% where $\mathbf{e}_i$ is a basis of the identity matrix with only the $i$-th element being $1$, and $\mathbf{P}^b = (P^b_{ij})$. Then, from \eqref{equ:delta-lambda} and Proposition \ref{pro:balance-lreig}, 
	% \begin{align*}
		%     \Delta_{\max} = \frac{-2P^b_{ij}y^b_ix^b_j - 2P^b_{ji}y^b_jx^b_i}{\mathbf{y}^{bT}\mathbf{x}^b} = \frac{-2\abs{W_{ij}} - 2\abs{W_{ij}}}{2m} = -\frac{2\abs{W_{ij}}}{m},
		% \end{align*}
	% where the second equality is obtained by $P^b_{ij}y_i^bx_j^b > 0$ and $P^v_{ji}y^b_jx^b_i > 0$ in balanced graphs.
	When $G^b$ can be obtained by flipping the sign of a set of edges $\tilde{E}\subset E$, 
	\begin{align*}
		\Delta\mathbf{P} = \sum_{(v_i,v_j)\in\tilde{E}}-2P^b_{ij}\mathbf{e}_i\mathbf{e}_j^T - 2P^b_{ji}\mathbf{e}_j\mathbf{e}_i^T,
	\end{align*}
	where $\mathbf{e}_i$ is a column vector of the identity matrix with only the $i$-th element being $1$, and $\mathbf{P}^b = (P^b_{ij})$. Then, from Corollary \ref{cor:balance-lreig}, 
	\begin{align*}
		\frac{\mathbf{w}^{bT}\Delta\mathbf{P}\mathbf{u}^b}{\mathbf{w}^{bT}\mathbf{u}^b}
		&= \frac{\sum_{(v_i,v_j)\in\tilde{E}}-2P^b_{ij}w^b_iu^b_j - 2P^b_{ji}w^b_ju^b_i}{\mathbf{w}^{bT}\mathbf{u}^b} \\
		&= \frac{\sum_{(v_i,v_j)\in\tilde{E}}-2\abs{W_{ij}} - 2\abs{W_{ij}}}{2m}\\ 
		&= -\frac{2\sum_{(v_i,v_j)\in\tilde{E}}\abs{W_{ij}}}{m},
	\end{align*}
	where the second equality is obtained by $P^b_{ij}w_i^bu_j^b > 0$ and $P^b_{ji}w^b_ju^b_i > 0$, $\forall (v_i,v_j)\in E$.
\end{proof}
Hence, from Proposition \ref{pro:P-perturb} together with \eqref{equ:delta-lambda}, the proposed measure $d_b(G) = -\Delta_{\max}$ is proportional to the number of edges disturbing the balanced structure, which is also known as frustration index or line index of balance \cite{abelson1958frustration,harary_1953_balance}. In a similar manner, we can show that the proposed measure $d_a(G) = \Delta_{\min}$ is proportional to the number of edges disturbing the antibalanced structure.

% \begin{proposition}
	%     If $\Delta \mathbf{P} = \mathbf{P} - \mathbf{P}^b$ is small, then the change of leading eigenvalue when adding one edge $(v_i, v_j)$ that is not consistent with the balanced structure is
	%     \begin{align}
		%         \Delta\lambda = -\frac{1}{m}\left(\frac{d_i}{d_i+1} + \frac{d_j}{d_j+1}\right).
		%     \end{align}
	%     Further, the change when adding a set of edges $\tilde{E} \subset E$ that are not consistent with the balanced structure is,
	%     \begin{align}
		%         \frac{1}{2m}\left( - \sum_{i:(i,j)\in \tilde{E}}\frac{|\tilde{E}_i|d_i}{k_i + |\tilde{E}_i|} + \sum_{i:(i,j)\in \tilde{E}}\frac{\sum_{j:(i,j)\in \tilde{E}_i}A'_{ij}v_iu_j}{d_i + |\tilde{E}_i|} \right).
		%     \end{align}
	% \end{proposition}

\subsubsection{Extended linear threshold (ELT) model}
\label{sec:elt}
In unsigned networks, the ELT model aggregates the state values from all the neighbours for each node, but only changes its state value if the sum is greater than a predetermined threshold. While in signed networks, we follow similar procedure to apply the opposing rule and consider the ``polarization" on each node as in subsection \ref{sec:linear_adj_dynamics}. We also maintain the thresholding, but further allow a node to take a negative value if the sum is less than the opposite of the threshold. Hence, the ELT model on signed networks evolves as follows, where $\forall v_j\in V, t > 0$,  
\begin{align}
	x_j(t) = 
	\begin{cases}
		\theta_{j,t},\quad & \sum_iW_{ij}x_i(t-1) \ge \theta_{j,t},\\
		-\theta_{j,t},\quad & \sum_iW_{ij}x_i(t-1) \le -\theta_{j,t},\\
		0,\quad &\text{otherwise},
	\end{cases}
	\label{equ:signed_extreme2_update}
\end{align}
where $\theta_{j,t}$ is the threshold to trigger the propagation, either positively or negatively, and the initial state vector $\mathbf{x}(0)$ is given.

A theoretical understanding of threshold models on simple networks is still an active area of research, and here we consider the even more challenging case of signed networks. Hence, we (i) start from a specific network structure, regular (ring) lattices with uniform magnitude of the edge weight, $\alpha$, and (ii) analyze the behavior of the ELT model when the whole neighbourhood of a node (including itself), referred to as the \textit{central node}, is activated. Specifically, we denote the degree of nodes in the regular lattices as $\bar{d}$, and apply the following geometric sequences (cf.~the threshold-type bounds \cite{Tian_info_2021}) for the threshold values
\begin{align*}
	\theta_{j,t} = (\theta_l\alpha)^tl_0,
\end{align*}
where $\theta_l>0$ is the time-independent threshold that is the same for all nodes, and $l_0>0$ is the magnitude of the initially activated state value, with $x_j(0) = \pm l_0$ if node $v_j$ is activated initially and $0$ otherwise. Then the updating function \eqref{equ:signed_extreme2_update} is now
\begin{align}
	x_j(t) = 
	\begin{cases}
		(\theta_l \alpha)^{t}l_0,\quad &\sum_{v_i\in \mathcal{A}_{t-1}}A_{ij} \ge \theta_l,\\
		-(\theta_l \alpha)^{t}l_0,\quad &\sum_{v_i\in \mathcal{A}_{t-1}}A_{ij} \le -\theta_l,\\
		0,\quad &\text{otherwise},
	\end{cases}
	\label{equ:extreme2_update-uni}
\end{align}
where $\mathbf{A}$ is the signed (unweighted) adjacency matrix as in \eqref{equ:signed_A}, and $\mathcal{A}_t = \{v_i: x_i(t) \ne 0\}$. Hence, $\theta_l$ is actually the threshold on the number of neighbours that are positively activated over those that are negatively activated (in the previous time step). 

We start from analyzing the ELT model on simple regular lattices (ignoring the edge sign), and then proceed to signed regular lattices through their balanced structures. In simple regular lattices, the whole neighbourhood of the central node are positively activated initially. We find that when  
\begin{align}
	\theta_l \le \theta_l^* = \bar{d}/2,
	\label{equ:lattice-upper}
\end{align}  
$\exists v_i\in V, t>0,\ s.t.\ x_i(t) > 0$, i.e., some node has positive state value at certain time step other than the initial start\footnote{\footnotesize{In this specific case, under the same condition, (i) every nodes will have positive state values at some time step, and (ii) only $d/2 + 1$ consecutive nodes are needed to be activated initially in order to trigger the propagation with feature (i).}}, and the condition is the same as the one in \cite{Centola_ccontagion_2007}. Specifically, at each $t>0$:
\begin{enumerate}[label=(\roman*)]
	\item $x_j(t) = (\theta_l \alpha)^tl_0$ if $v_j\in \mathcal{A}_{t-1}$;
	\item $\bar{d} - 2(\left \lceil{\theta_l}\right \rceil - 1)$ more nodes that are closest to $\mathcal{A}_{t-1}$ will be activated with the same state value $(\theta_l \alpha)^tl_0$, if there is any.
\end{enumerate}
Therefore, $x_i(t) > 0,\ \forall v_i\in V$, for sufficiently large $t$. In the following, we will specify the behavior of the ELT model on signed regular lattices that are balanced, antibalanced or strictly unbalanced. Here, $x_i(t) < 0$ is possible for each node $v_i$ at each time step $t\ge 0$, and particularly we will specify whether to positively or negatively activate a node initially. 

\paragraph{Balanced regular lattices.} We consider the activations that are consistent with the balanced bipartition, where we positively activate the central node $v_i$, and for each of its neighbours $v_j$, we positively activate it with $x_j(0) = l_0$ if $W_{ij} > 0$ and negatively activate it with $x_j(0) = -l_0$ otherwise; see \eqref{fig:lattice} for different distributions of signed edges and the corresponding activations.
\begin{figure}[ht]
	\centering
	\begin{tabular}{cc}
		\includegraphics[width=.3\textwidth]{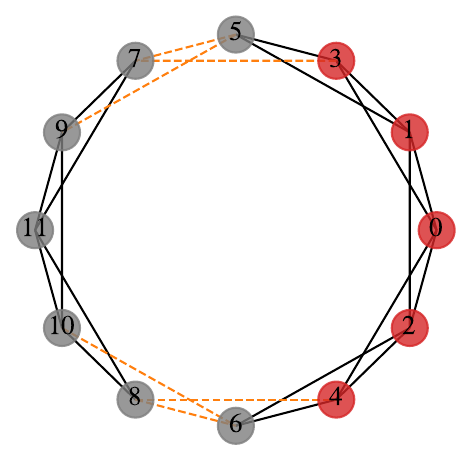} & \includegraphics[width=.3\textwidth]{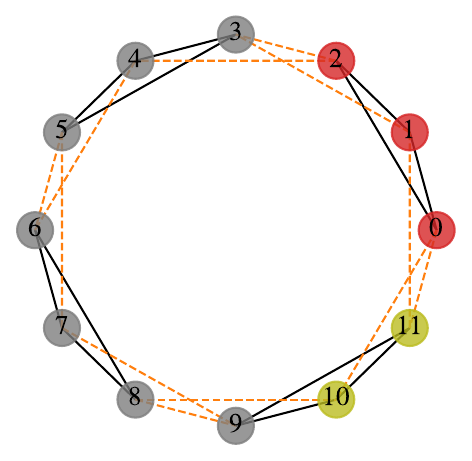}
	\end{tabular}
	\caption{Example of signed regular lattices of degree $4$ with different structurally balanced configurations, where positive edges are in black, negative edges are dashed in orange, and the whole neighbourhood of node $v_0$ is in different colour(s) from the others (in grey), with the ones that are positively activated in red and the others that are negatively activated in green.}
	\label{fig:lattice}
\end{figure}

Similarly, we find that when condition \eqref{equ:lattice-upper} is true, $\exists v_i\in V, t>0,\ s.t.\ x_i(t) \ne 0$, i.e., some node has nonzero state value at certain time step other than the initial start, and we refer to this phenomenon as ``certain propagation" on the signed networks. However, with the edge sign, there are more interesting patterns, where at each $t > 0$: 
\begin{enumerate}[label=(\roman*)]
	\item 
	\begin{align*}
		x_j(t) = 
		\begin{cases}
			(\theta_l \alpha)^tl_0, \quad \text{if } v_j\in \mathcal{A}_{t-1}^+,\\
			-(\theta_l \alpha)^tl_0, \quad \text{if } v_j\in \mathcal{A}_{t-1}^-,
		\end{cases}
	\end{align*}
	where $\mathcal{A}_{t} = \mathcal{A}_{t}^+\cup\mathcal{A}_{t}^-$ with $\mathcal{A}_{t}^+ = \{v_j:x_j(t) > 0\}$ and $\mathcal{A}_{t}^- = \{v_j:x_j(t) < 0\}$;
	\item $\bar{d} - 2(\left \lceil{\theta_l}\right \rceil - 1)$ more nodes that are closest to $\mathcal{A}_{t-1}$ will be activated if there is any, where
	\begin{align*}
		x_j(t) = 
		\begin{cases}
			(\theta_l \alpha)^tl_0, \quad \text{if } \exists v_{i_1}\in \mathcal{A}_{t-1}^+ \text{ or } v_{i_2}\in \mathcal{A}_{t-1}^- \text{ s.t. } A_{i_1 j} > 0 \text{ or } A_{i_2 j} < 0,\\
			-(\theta_l \alpha)^tl_0, \quad \text{if } \exists v_{i_1}\in \mathcal{A}_{t-1}^+ \text{ or } v_{i_2}\in \mathcal{A}_{t-1}^- \text{ s.t. } A_{i_1 j} < 0 \text{ or } A_{i_2 j} > 0. 
		\end{cases}
	\end{align*}
\end{enumerate}
Hence, each $x_j(t)$ has the same magnitude as the state value on the corresponding simple regular lattice. For the sign pattern, $\mathcal{A}^+_{t-1}\subset \mathcal{A}^+_{t}$ and $\mathcal{A}^-_{t-1}\subset \mathcal{A}^-_{t}$, $\forall t > 0$, i.e., the nodes, once activated, remain active and maintain the sign of their state values over time, which is similar to the evolution of $\mathbf{W}^t$ in the linear adjacency dynamics when the underlying signed network is balanced. 

\paragraph{Antibalanced regular lattices.} Corresponding to the analysis in balanced regular lattices, we consider the activations that are consistent with the antibalanced bipartition, where we positively activate the central node $v_i$, and for each of its neighbour $v_j$, we positively activate it with $x_j(0) = l_0$ if $W_{ij} < 0$ and negatively activate it otherwise; see Figure \ref{fig:lattice-anti} for different distributions of signed edges and the corresponding activations. 
\begin{figure}[ht]
	\centering
	\begin{tabular}{cc}
		\includegraphics[width=.3\textwidth]{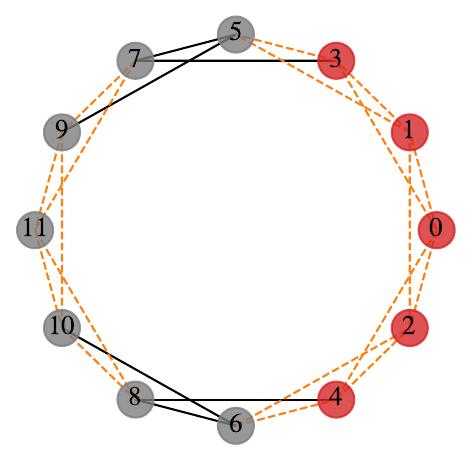} & \includegraphics[width=.3\textwidth]{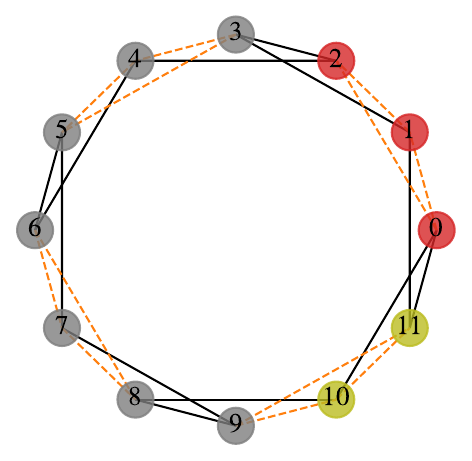}
	\end{tabular}
	\caption{Example of signed regular lattices of degree $4$ with different structurally antibalanced configurations, where positive edges are in black, negative edges are dashed in orange, and the whole neighbourhood of node $v_0$ are in different colour(s) from the others (in grey), with the ones that are positively activated in red and the others that are negatively activated in green.}
	\label{fig:lattice-anti}
\end{figure}

Again, we find that when condition \eqref{equ:lattice-upper} is true, there is certain propagation on the structurally antibalanced regular lattice, but an alternating sign pattern. Specifically, at each $t > 0$: 
\begin{enumerate}[label=(\roman*)]
	\item 
	\begin{align*}
		x_j(t) = 
		\begin{cases}
			-(\theta_l \alpha)^tl_0, \quad \text{if } v_j\in \mathcal{A}_{t-1}^+,\\
			(\theta_l \alpha)^tl_0, \quad \text{if } v_j\in \mathcal{A}_{t-1}^-,
		\end{cases}
	\end{align*}
	where $\mathcal{A}_{t} = \mathcal{A}_{t}^+\cup\mathcal{A}_{t}^-$ with $\mathcal{A}_{t}^+ = \{v_j:x_j(t) > 0\}$ and $\mathcal{A}_{t}^- = \{v_j:x_j(t) < 0\}$;
	\item $\bar{d} - 2(\left \lceil{\theta_l}\right \rceil - 1)$ more nodes that are closest to $\mathcal{A}_{t-1}$ will be activated if there is any, where
	\begin{align*}
		x_j(t) = 
		\begin{cases}
			(\theta_l \alpha)^tl_0, \quad \text{if } \exists v_{i_1}\in \mathcal{A}_{t-1}^+ \text{ or } v_{i_2}\in \mathcal{A}_{t-1}^- \text{ s.t. } A_{i_1 j} > 0 \text{ or } A_{i_2 j} < 0,\\
			-(\theta_l \alpha)^tl_0, \quad \text{if } \exists v_{i_1}\in \mathcal{A}_{t-1}^+ \text{ or } v_{i_2}\in \mathcal{A}_{t-1}^- \text{ s.t. } A_{i_1 j} < 0 \text{ or } A_{i_2 j} > 0. 
		\end{cases}
	\end{align*}
\end{enumerate}
Hence again, each $x_j(t)$ has the same magnitude as the state value on the corresponding simple regular lattice. However, for the sign pattern, $\mathcal{A}^+_{t-1}\subset \mathcal{A}^-_{t}$ and $\mathcal{A}^-_{t-1}\subset \mathcal{A}^+_{t}$, i.e., the nodes, once activated, remain active but alternate the sign of their state values in every time step, which is similar to the evolution of $\mathbf{W}^t$ in the linear adjacency dynamics when the underlying signed network is antibalanced. 

\paragraph{Strictly unbalanced regular lattices.} In all the remaining configurations, neither are they balanced where the nodes, once activated, remain active and maintain the sign of their state values over time, nor are they antibalanced where the nodes, once activated, remain active but alternate the sign of their state values in every time step. There could be conflicts in the sign of a node's neighbours' state values multiplying the edge weights, hence it is more likely for the sum to be less than the threshold, and for these nodes to have state value $0$ accordingly. Hence, the propagation in strictly unbalanced lattices generally terminates within less number of time steps. 

We can still find the same condition \eqref{equ:lattice-upper} on $\theta_l$ to trigger certain propagation on strictly unbalanced regular lattices, but generally not all nodes will have nonzero state values in the propagation process. The behavior over time could be more similar to balanced or antibalanced regular lattices, depending on how far it is from being balanced by \eqref{equ:signed_d-balanced} compared with that from being antibalanced by \eqref{equ:signed_d-anti-balanced}.

\paragraph{General signed networks.} In this section, we have analyzed the behavior of the ELT model on signed regular lattices, from the perspective of the balanced structure. For general signed networks, we can consider the performance of ELT model as follows. (i) We interpolate the signed network locally by signed regular lattices of different degrees, or signed trees where complex contagions (e.g., $\theta_l > 1$ on the signed networks with uniform magnitude of the edge weight) can hardly proceed. (ii) Then we can estimate the behavior of the ELT model on the whole network by interpolating that on the corresponding signed regular lattices.

\section{Numerical experiments}
\label{sec:numerical_exp}
In this section, we numerically explore the dynamics on signed networks, and verify the results we have shown for structurally balanced, antibalanced and strictly unbalanced networks. Specifically, we consider both linear models, the linear adjacency dynamics and signed random walks, and a nonlinear one, the ELT model, and illustrate their consistent patterns, on both synthetic networks generated from the signed stochastic block model, and also a real signed network of Highland tribes.

% \begin{itemize}
	%     \item The linear adjacency dynamics, including the signed random walks. 
	%     \item The extended linear threshold model. 
	% \end{itemize}

\begin{figure}[ht]
	\centering
	\begin{tabular}{cc}
		\includegraphics[width=.4\textwidth]{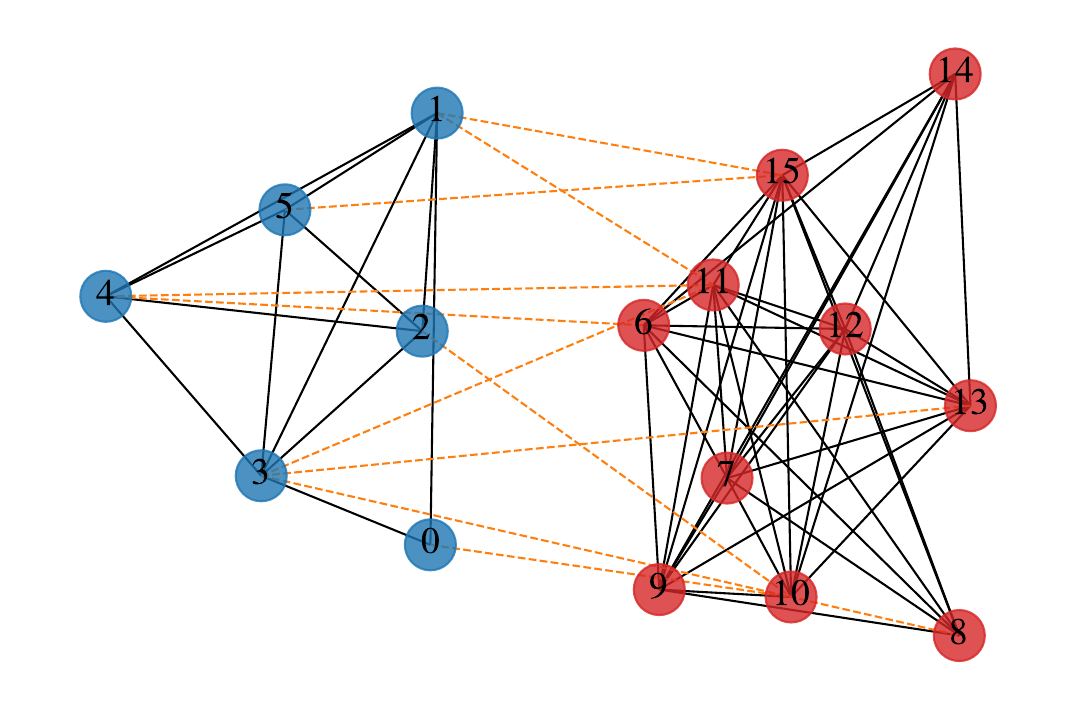} & \includegraphics[width=.4\textwidth]{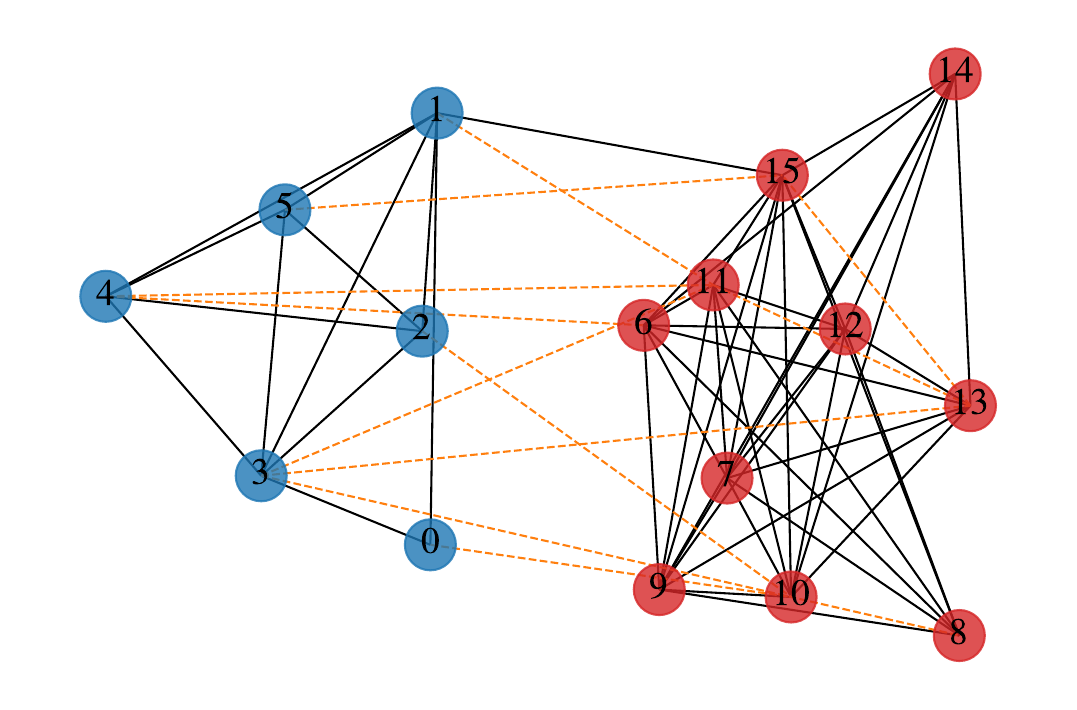} \\
		\includegraphics[width=.4\textwidth]{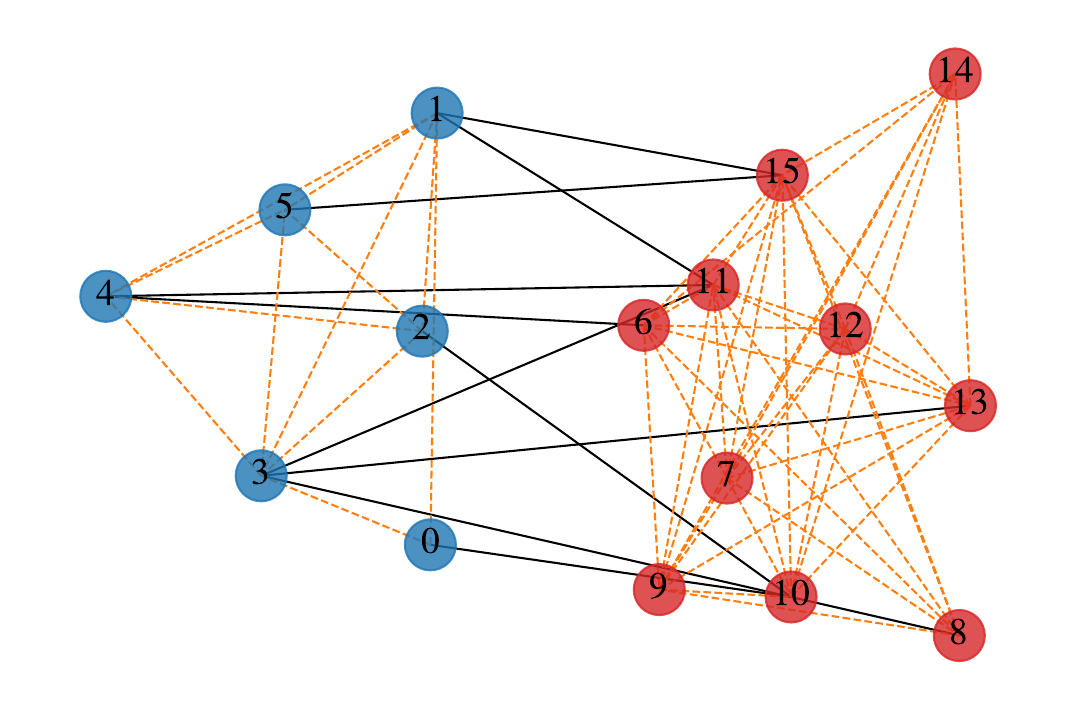} & \includegraphics[width=.4\textwidth]{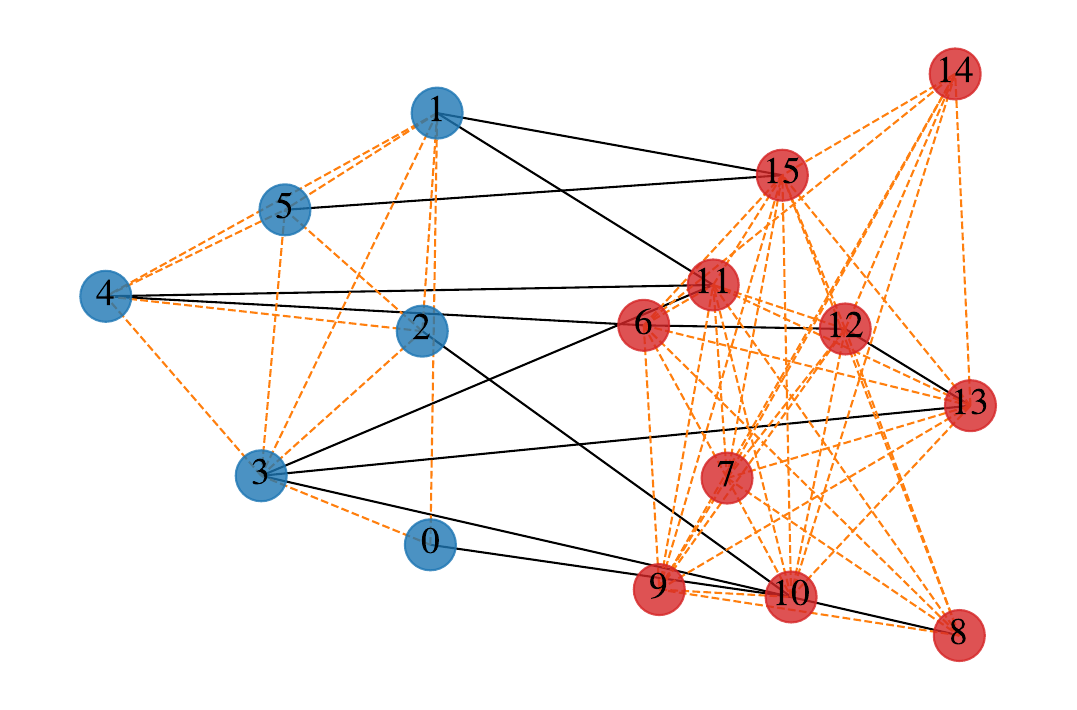}
	\end{tabular}
	\caption{Signed networks from SSBM that are balanced ($\eta = 0$, upper left), close to being balanced ($\eta = 0.05$, upper right) with $d_b(G) = 0.077,\, d_a(G) = 0.503$ and $3$ edges disturbing the balanced structure, antibalanced ($\eta=1$, bottom left), and close to being antibalanced ($\eta = 0.95$, bottom right) with $d_b(G) = 0.525,\, d_a(G) = 0.047$ and $2$ edges disturbing the antibalanced structure, where the node colour indicates the bipartitions of the relevant balanced or antibalanced structures, and the edge colour indicates the sign (black: positive; orange: negative).}
	\label{fig:exp-ssbms}
\end{figure}
%$\lambda_{\max}(\mathbf{P}) = 0.923, \lambda_{\min}(\mathbf{P}) = -0.497$
%$\lambda_{\max}(\mathbf{P}) = 0.475, \lambda_{\min}(\mathbf{P}) = -0.963$
\paragraph{Signed stochastic block model (SSBM).} We consider the signed stochastic block model (SSBM) which is slightly modified from the version in \cite{Cucuringu_SPONGE_2019}. The SSBM is constructed by the following components: (i) a planted SBM, $SBM(p_{in}, p_{out})$, where the probabilities of an edge to occur inside each community and between the two communities being $p_{in}$ and $p_{out}$, respectively; (ii) an initial balanced configuration, where edges inside each community being positive while those between the two communities are negative; (iii) a probability $\eta\in[0,1]$\footnote{\footnotesize{The original range is $[0,1/2)$ in \cite{Cucuringu_SPONGE_2019} to obtain perturbations from (at least) weakly balanced networks. Here we extend it to be the whole range because networks close to being antibalanced (in the two-block case) could also be interesting.}} to flip the edge sign. We denote this planted SSBM by $SSBM(p_{in}, p_{out}, \eta)$. We start from a network of size $n=16$, two blocks with $n_1 = 6$ in one and $n_2 = 10$ in the other, and $p_{in} = 0.8$, $p_{out} = 0.1$. The choice of network size is for visualisation purposes, where for networks of larger sizes, the results will be very similar if we maintain the network density. We consider the following cases of the flipping probability, $\eta = 0$ for a balanced signed network\footnote{\footnotesize We note that $\eta=0$ corresponds to balanced signed networks because there are two blocks, and also that the two blocks are the corresponding bipartition. The case for antibalanced signed networks is similar.}, $\eta = 0.05$ for a signed network close to being balanced, $\eta = 1$ for an antibalanced signed networks, and $\eta = 0.95$ for a network close to being antibalanced; see 
Figure \ref{fig:exp-ssbms}. In all above cases, we assign a uniform magnitude $\alpha=0.1$ to the edge weights.
% \begin{figure}[ht]
	%     \centering
	%     \includegraphics[width=.48\textwidth]{figures/exp/SSBM-pf0_degs_init6.pdf}
	%     \caption{The degree of nodes in the signed networks.}
	%     \label{fig:ssbm-deg}
	% \end{figure}

For the linear adjacency dynamics, we observe that the state values are positive in one node subset of the bipartition and negative in the other in the balanced one, while alternate their signs in the antibalanced one; see Figure \ref{fig:ssbm-linear}. The evolution of state values of the signed network close to being balanced is very similar to the balanced one, while the other close to being antibalanced has very similar performance to the antibalanced one. Meanwhile, we can already observe in the first few steps that some nodes have state value of smaller magnitudes in the strictly unbalanced networks than either the balanced or the antibalanced ones, such as node $v_{15}$ in the upper right case of Figure \ref{fig:ssbm-linear}. See Figure \ref{fig:ssbm-linear-t} in Appendix \ref{sec:app_proofs} for further comparison to elucidate the differences in more time steps.
\begin{figure}[ht]
	\centering
	\hspace*{-2em}
	\begin{tabular}{cc}
		\includegraphics[width=.48\textwidth]{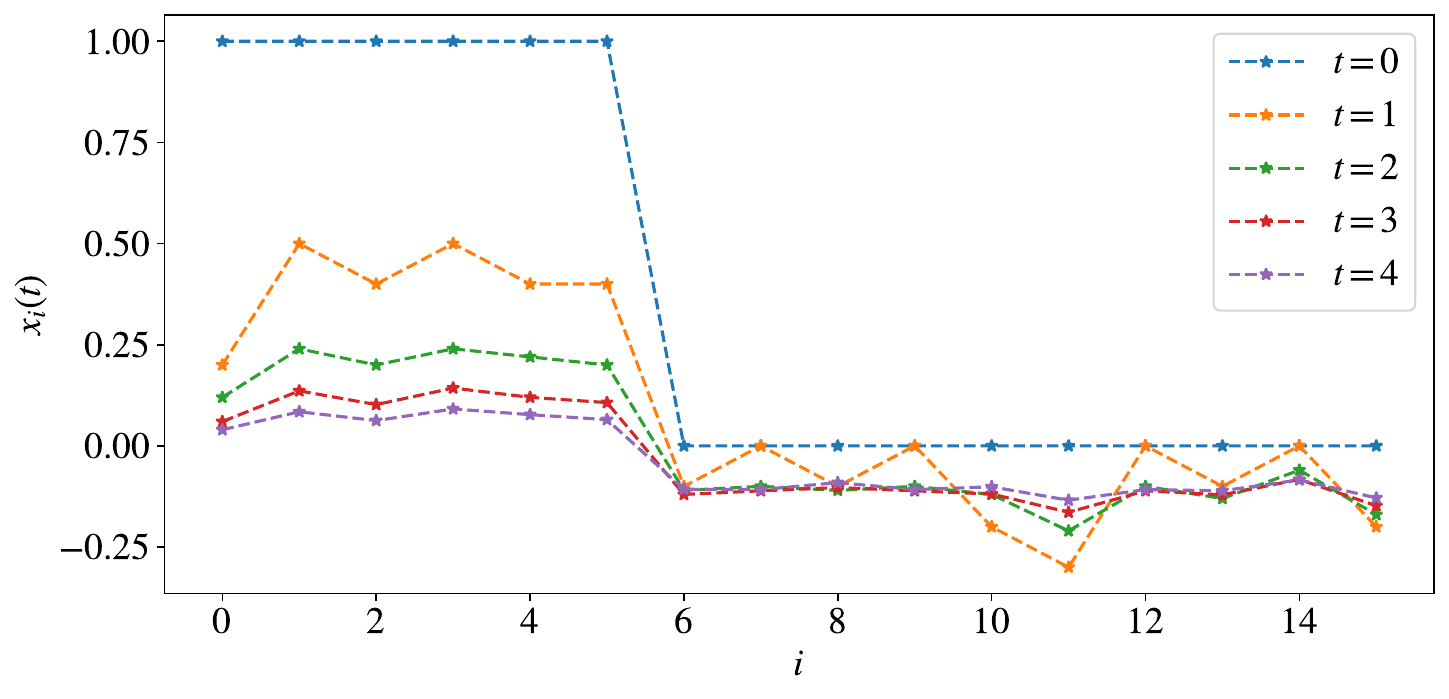} & \includegraphics[width=.48\textwidth]{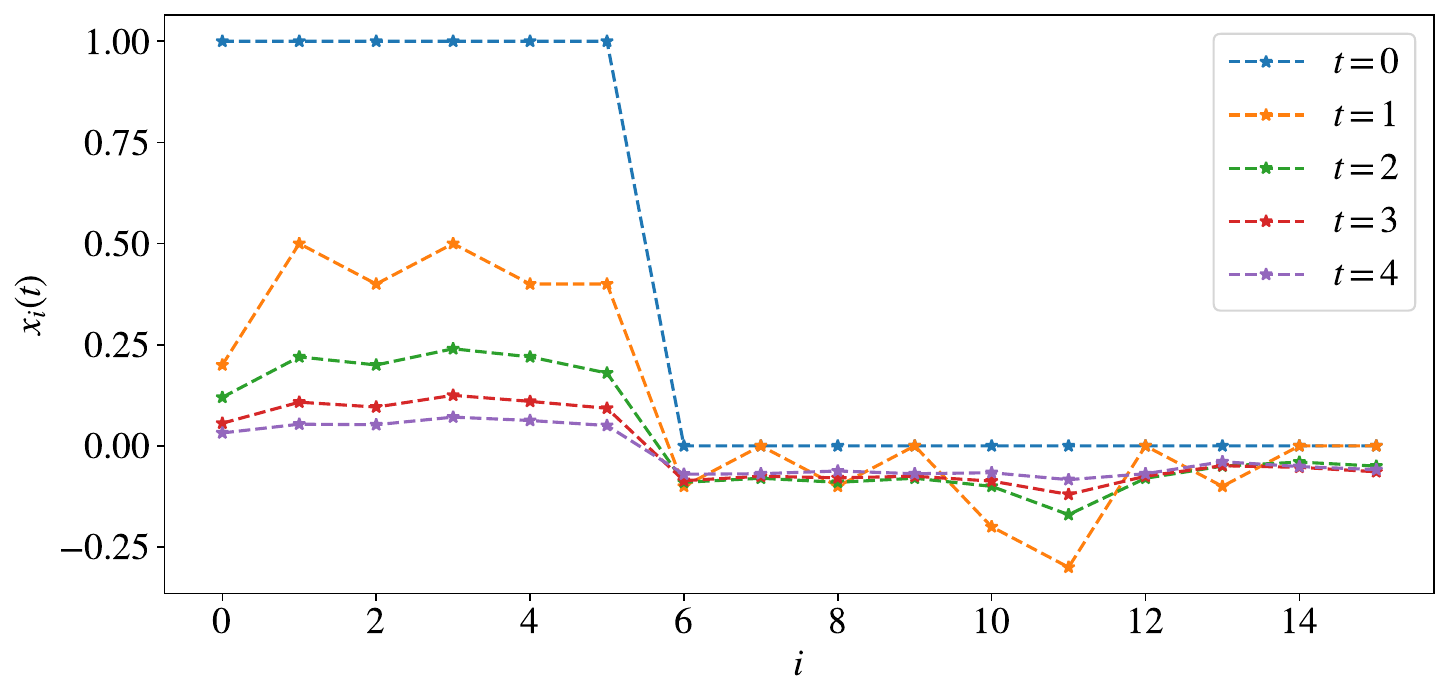} \\
		\includegraphics[width=.48\textwidth]{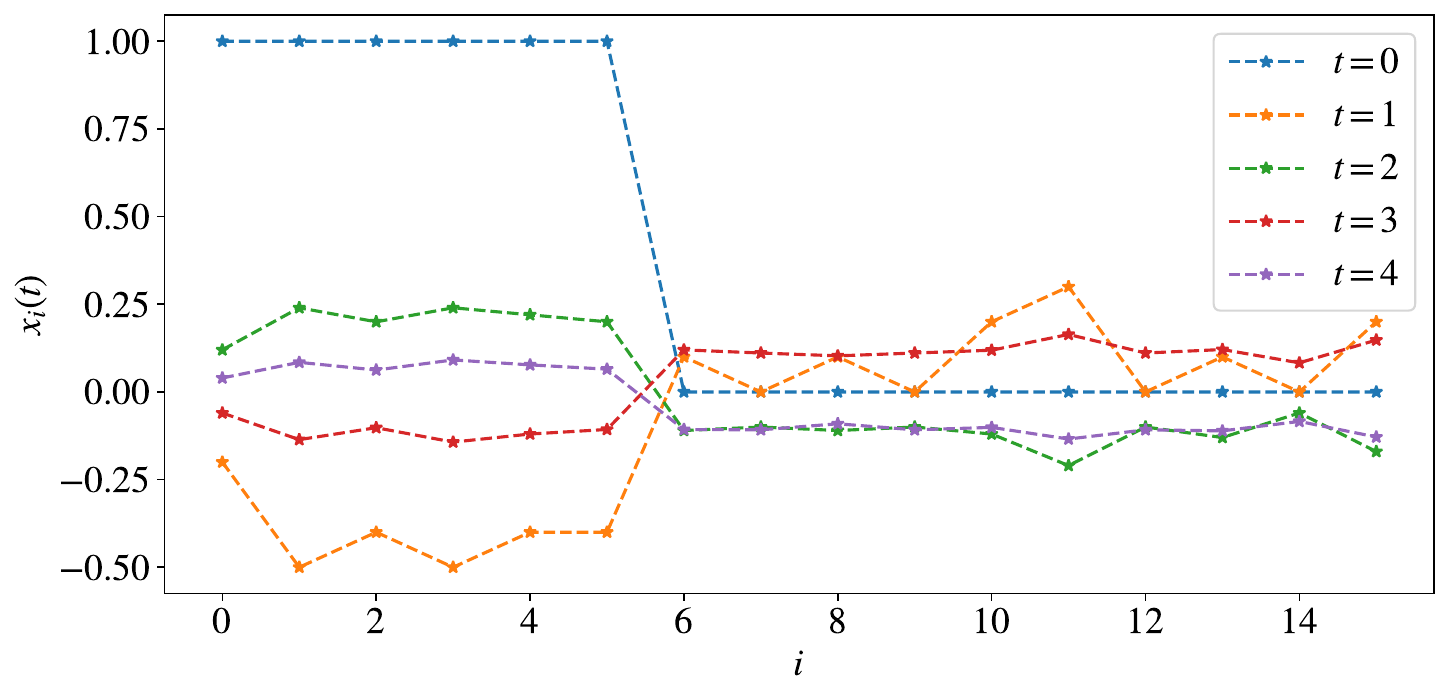} & \includegraphics[width=.48\textwidth]{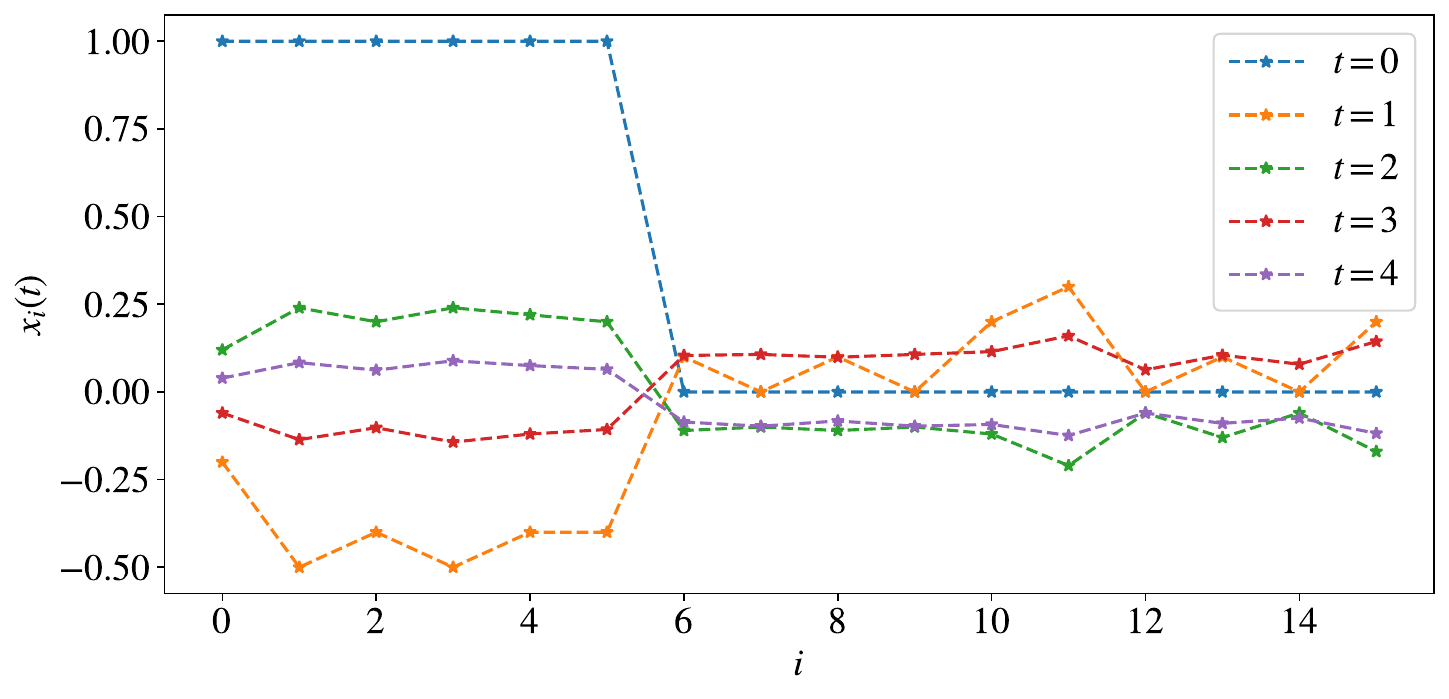}
	\end{tabular}
	\caption{Evolution of the state values from the linear adjacency dynamics on the networks in Figure \ref{fig:exp-ssbms} that are balanced ($\eta = 0$, upper left), close to being balanced ($\eta = 0.05$, upper right), antibalanced ($\eta=1$, bottom left), and close to being antibalanced ($\eta = 0.95$, bottom right), where, in the same row, the results on the right are expected to be close to those on the left, to some extent, since the underlying networks are close to each other by our proposed measures.}
	\label{fig:ssbm-linear}
\end{figure}

For the signed random walks, we observe similar behaviors as in the case of linear adjacency dynamics, where the state values are positive in one node subset of the bipartition while negative in the other (asymptotically) in the balanced network, while alternate their signs in the antibalanced one; see Figure \ref{fig:ssbm-rw}. However for random walks, the differences between being exactly balanced or close to balanced are clearer, since the state values in the latter will eventually converge to $0$ while the former will have nonzero state values asymptotically. The case of the antibalanced one and another that is close to being antibalanced is similar. Furthermore, we find the steady state values are proportional to the node degree in the balanced or antibalanced networks, which are consistent with Propositions \ref{pro:balance-steady} and \ref{pro:antibalance-steady}. 
\begin{figure}[ht]
	\centering
	\hspace*{-2em}
	\begin{tabular}{cc}
		\includegraphics[width=.48\textwidth]{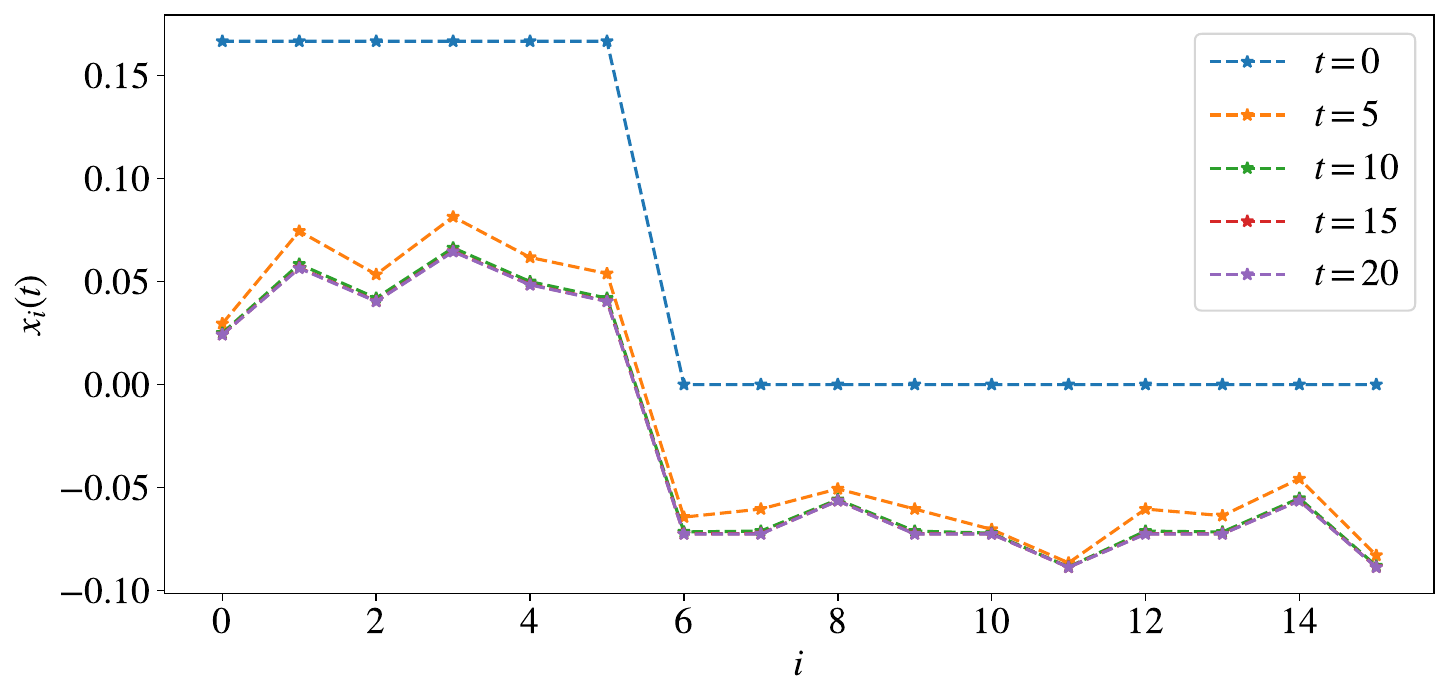} & \includegraphics[width=.48\textwidth]{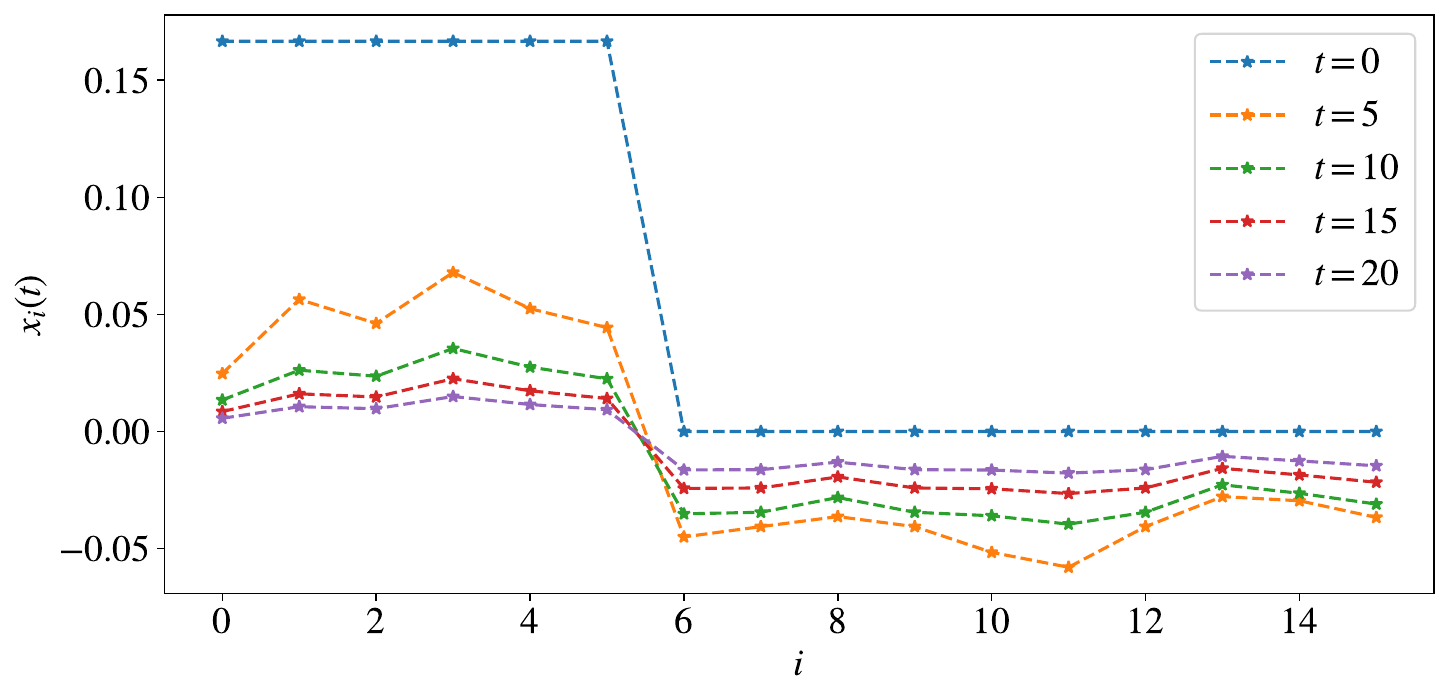} \\
		\includegraphics[width=.48\textwidth]{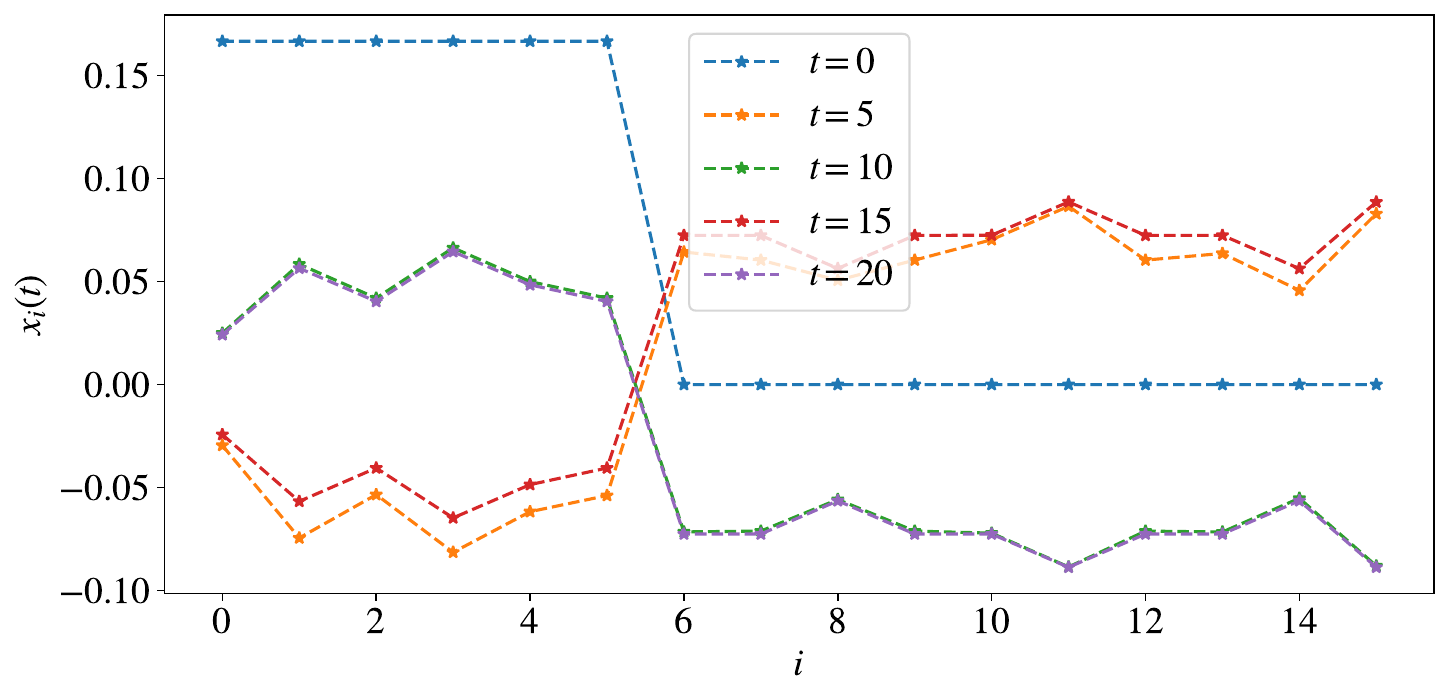} & \includegraphics[width=.48\textwidth]{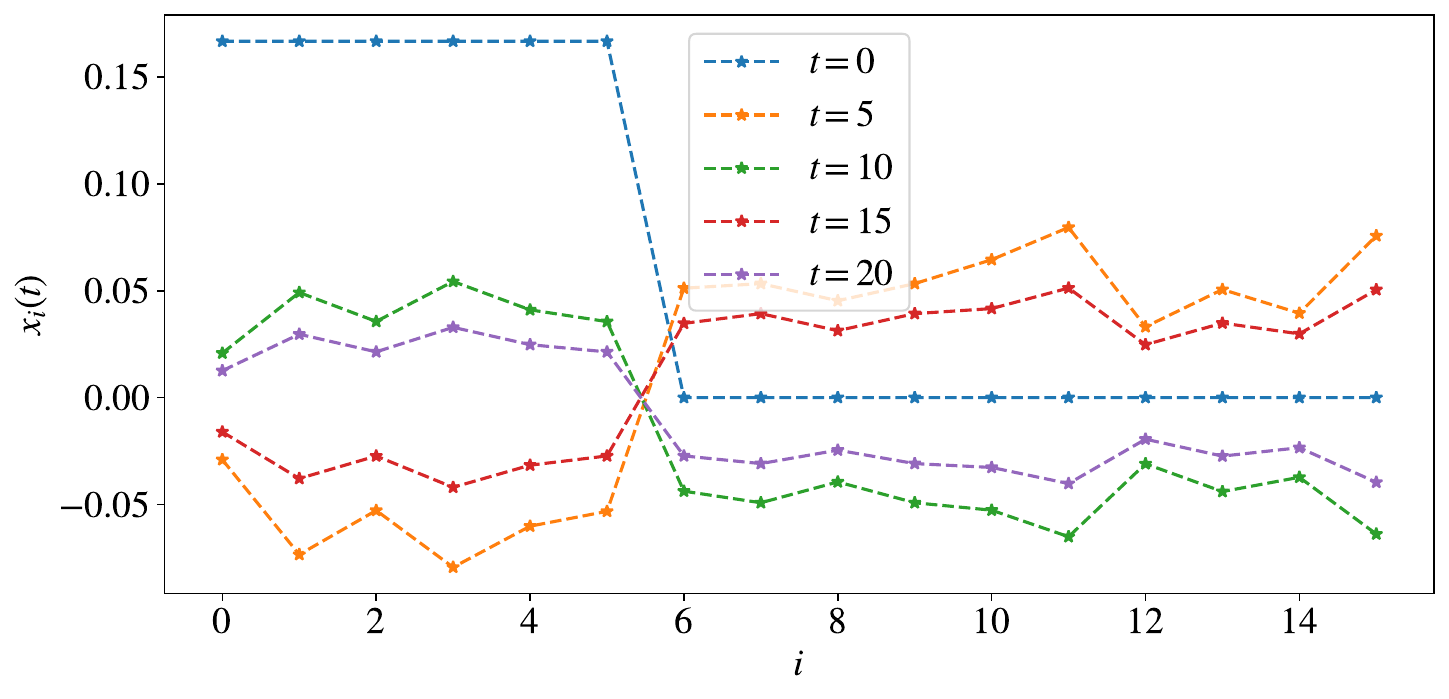}
	\end{tabular}
	\caption{Evolution of the state values from the signed random walks on the networks in Figure \ref{fig:exp-ssbms} that are balanced ($\eta = 0$, upper left), close to being balanced ($\eta = 0.05$, upper right), antibalanced ($\eta=1$, bottom left), and close to being antibalanced ($\eta = 0.95$, bottom right), where the difference between the results in the same row is expected to be larger in this case.}
	\label{fig:ssbm-rw}
\end{figure}

For the ELT model, even though being very different from the previous linear adjacency dynamics, we observe similar signed patterns, where the state values are positive in one node subset of the bipartition and negative in the other in the balanced network, while alternate their signs in the antibalanced one; see Figure \ref{fig:ssbm-LT}. The evolution of state values of the signed network close to being balanced (antibalanced) is, again, very similar to the balanced (antibalanced) one, while we can also observe that some nodes in the strictly unbalanced networks already have state values of smaller magnitudes in the first few steps; see again node $v_{15}$ in the upper right of Figure \ref{fig:ssbm-LT}.
\begin{figure}[ht]
	\centering
	\hspace*{-2em}
	\begin{tabular}{cc}
		\includegraphics[width=.48\textwidth]{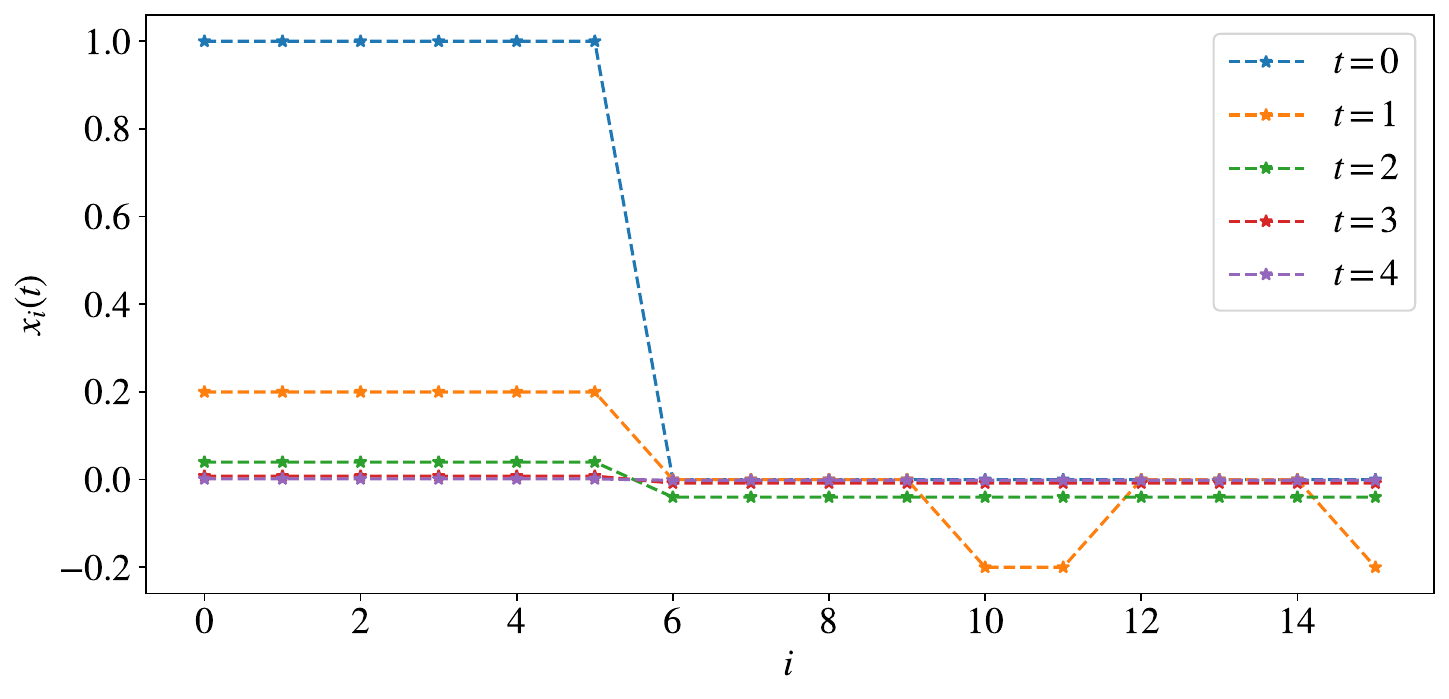} & \includegraphics[width=.48\textwidth]{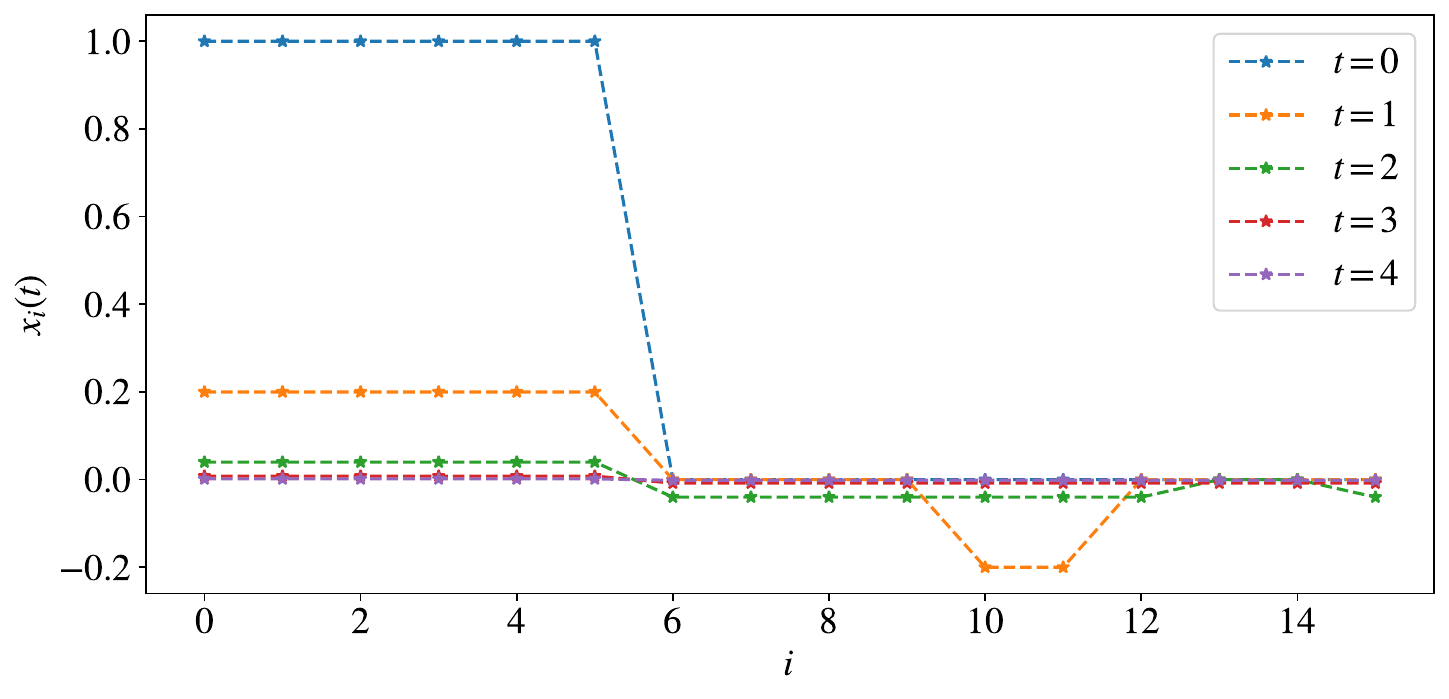} \\
		\includegraphics[width=.48\textwidth]{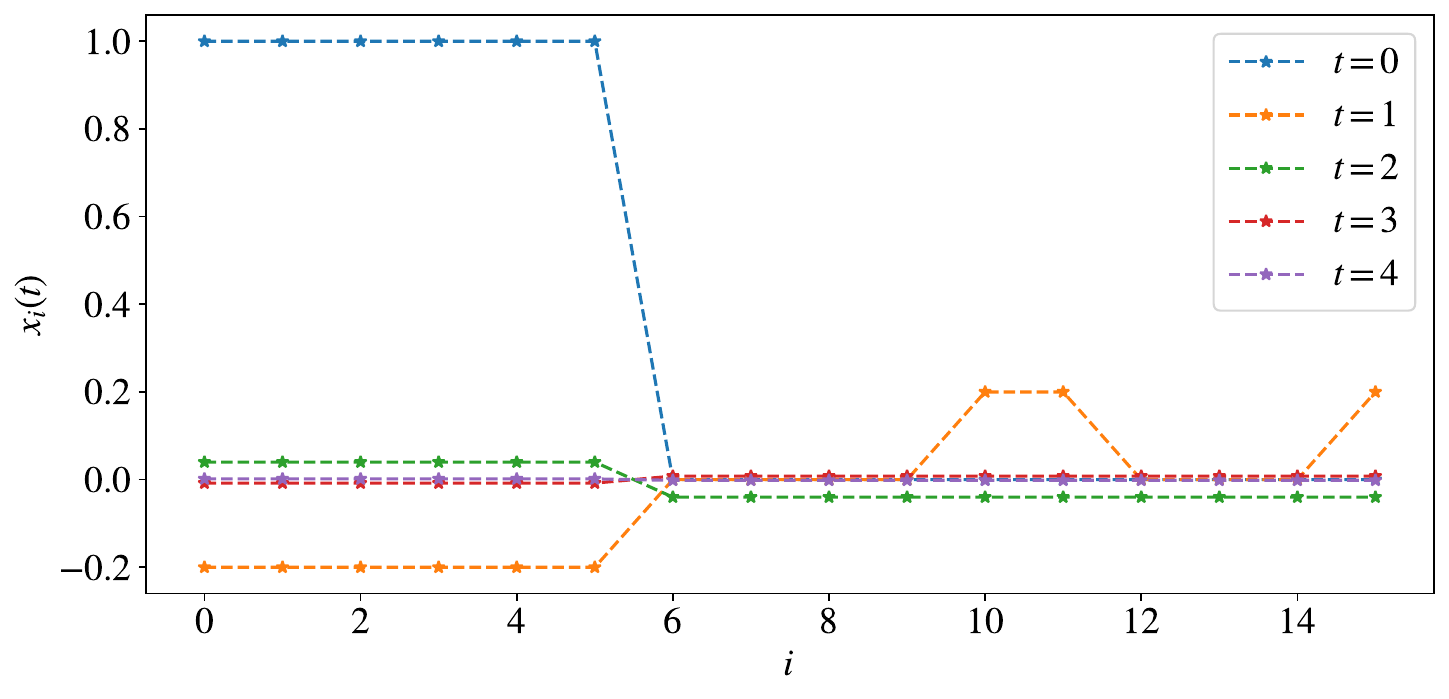} & \includegraphics[width=.48\textwidth]{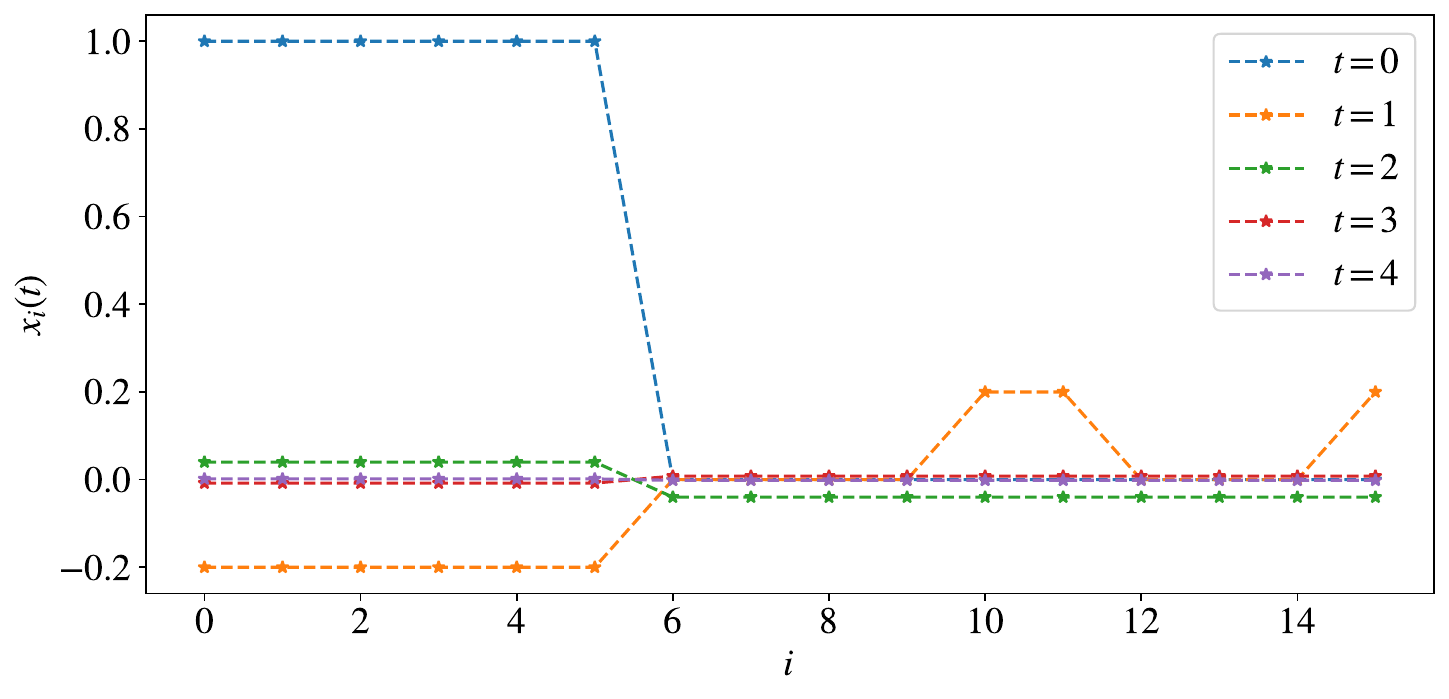}
	\end{tabular}
	\caption{Evolution of the state values from the ELT model on the networks in Figure \ref{fig:exp-ssbms} that are balanced ($\eta = 0$, upper left), close to being balanced ($\eta = 0.05$, upper right), antibalanced ($\eta=1$, bottom left), and close to being antibalanced ($\eta = 0.95$, bottom right).}
	\label{fig:ssbm-LT}
\end{figure}

\paragraph{Real signed networks.} We now consider a classic example of real signed networks, the Highland tribes network. The network contains $n=16$ tribes as nodes, and a positive edge indicates friendship while a negative edge indicates enmity. Hage and Harary \cite{Hage_model_1983} used the Gahuku-Gama system of the Eastern Central Highlands of New Guinea, described by Read \cite{Read_NewG_1954}, to illustrate a clusterable signed network. This network is known to be weakly balanced, where the enemy of an enemy can be either a friend or an enemy, thus is neither balanced nor antibalanced. We find $d_b(G) = 0.155$ while $d_a(G) = 0.529$, thus the network is relatively closer to being balanced; see Figure \ref{fig:exp-tribes}. Here, we assign a uniform magnitude $\alpha=0.1$ to the edge weights.
\begin{figure}[ht]
	\centering
	\includegraphics[width=.48\textwidth]{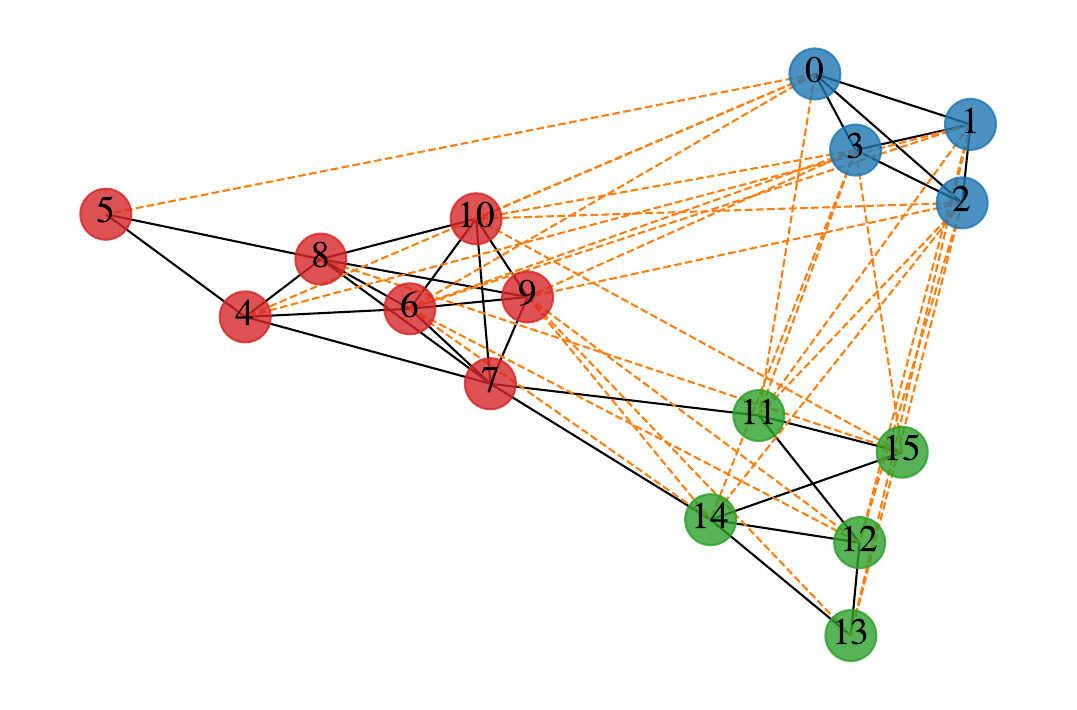}
	\caption{Highland tribes network with $d_b(G) = 0.155$ and $d_a(G) = 0.529$, where the node colour indicates the partition of the network with only positive edges inside, since it is closer to being balanced than antibalanced, and the edge colour indicates the sign (black: positive; orange: negative).}
	\label{fig:exp-tribes}
\end{figure}
%$\lambda_{\max}(\mathbf{P}) = 0.845, \lambda_{\min}(\mathbf{P}) = -0.471$

We then explore the dynamics, both the linear one from random walks and the nonlinear one from the ELT model, when activating the nodes in the red part (i.e., node $\{v_4,\dots, v_{10}\}$) in Figure \ref{fig:exp-tribes}. We find consistent signed patterns in these two dynamics, where nodes in the red part remain positive states and nodes in the blue part remain negative states. While for nodes in the green part, there is more variance, since there are both positive and negative edges connecting them and the red part, where node $v_{11}$ has positive states, nodes $v_{12},\, v_{15}$ have negative states, while for the remaining nodes, their state values quickly approach $0$; see Figure \ref{fig:tribe-rwlt}. We still observe that the state values in random walks reach zero at stationary. 
\begin{figure}[ht]
	\centering
	\hspace*{-2em}
	\begin{tabular}{cc}
		\includegraphics[width=.48\textwidth]{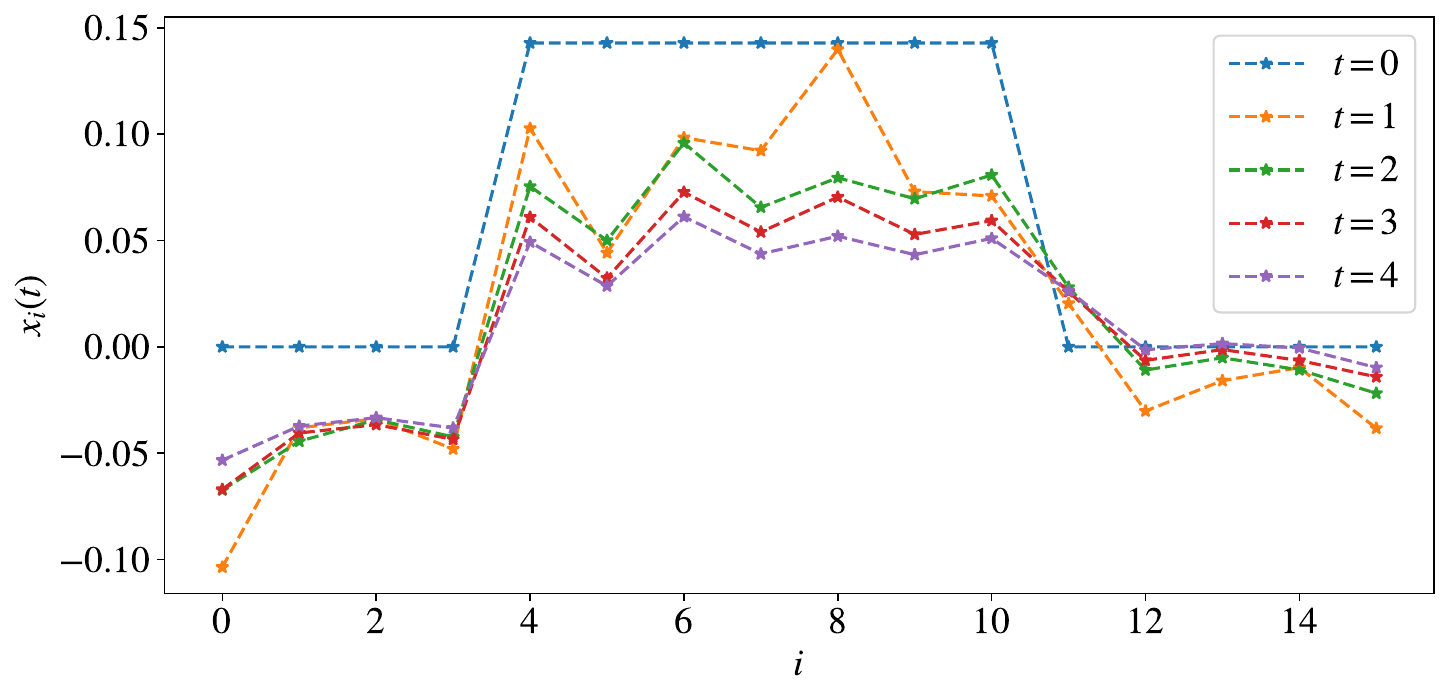} & \includegraphics[width=.48\textwidth]{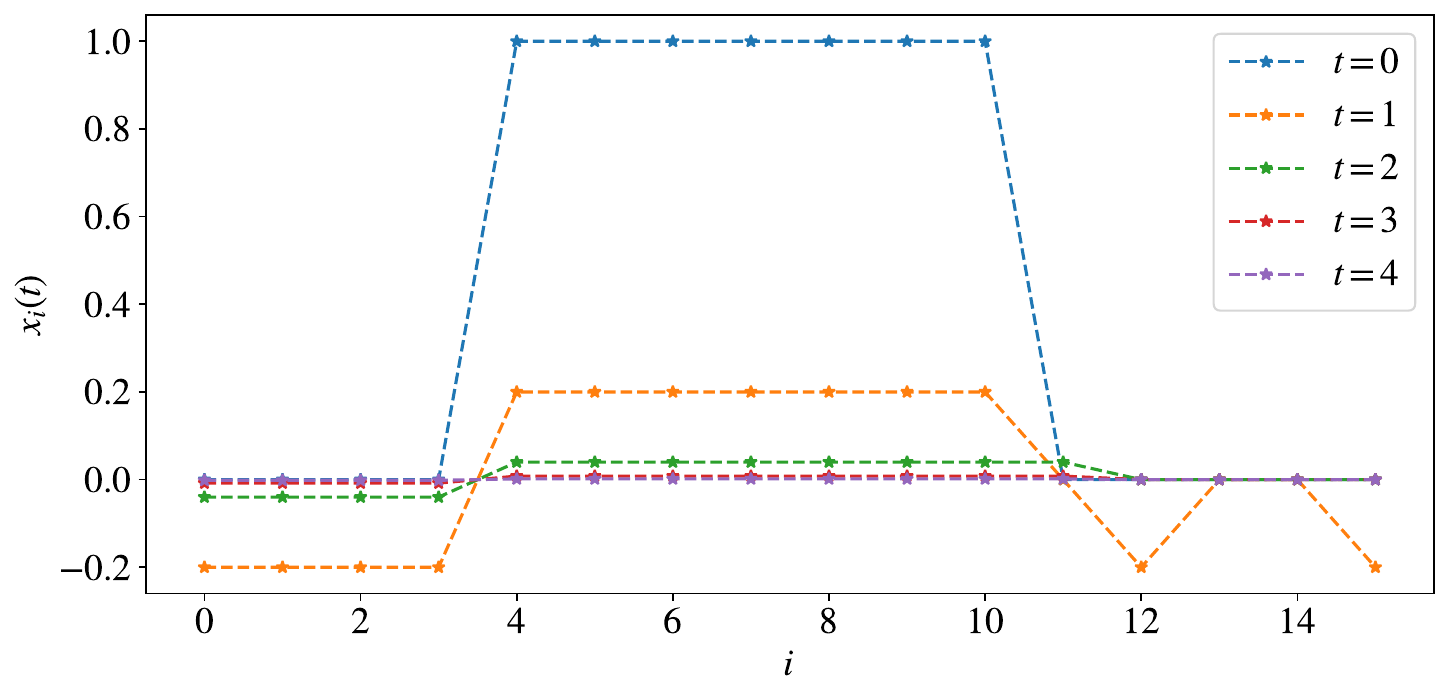} 
	\end{tabular}
	\caption{Evolution of the state values from the signed random walks (left) and the ELT model (right) on the Highland tribes network.}
	\label{fig:tribe-rwlt}
\end{figure}

\section{Conclusions}
\label{sec:conclusions}
% \begin{itemize}
	%     \item Signed networks are interesting. 
	%     \item Structural balance can help us understand the spreading on signed networks.
	%     \item Random walks on signed networks are interesting, and could be used for designing partition algorithms. 
	% \end{itemize}
Signed networks have gained increasing popularity over recent years, due to their extensive applications in various domains. Most of existing results  are obtained for unweighted signed networks, and the dynamics are usually characterized exclusively on structurally balanced networks, while signed networks can be weighted and unbalanced. In this paper, our focus is on the whole range of signed networks. We first classify  signed networks, where one type corresponds to structurally balanced ones on which most literature focused, while there are also two more types to further specify the unbalanced networks, i.e., structurally antibalanced and strictly unbalanced ones. Then to characterize each type, we consider the spectral properties of the weighted adjacency matrix. In particular, we show that the spectral radius contracts with the presence of signs if and only if the signed network is strictly unbalanced. Based on this classification  and  on  further structural analysis, we show that linear models, such as the linear adjacency dynamics, and nonlinear models, such as the ELT model, can have similar patterns over time. Specifically, subject to appropriate initialization, for balanced networks, the state values maintain the same sign over time, while for antibalanced networks, the state values alternate the sign in every step, and for strictly unbalanced networks, the state values can have smaller magnitude than those obtained on the network ignoring the sign. As an example, we show that signed random works have a zero vector as the steady state if and only if the signed network is strictly unbalanced. The numerical results from the synthetic networks generated from SSBM, and the real Highland tribes network further confirm our findings. 

In this work, we have provided a full description of the eigenpairs of the weighted adjacency matrix when the signed networks are balanced or antibalanced, while for the strictly unbalanced networks, we can only characterize the spectral radius. We note that structurally balanced and antibalanced networks only consist of two switching equivalence classes of signed networks \cite{Atay_signedCheeger_2020}, and there are potentially more signed networks that lie in the broad category of strictly unbalanced networks. As a next step, it would be interesting to further characterize the signed networks that are strictly unbalanced in terms of switching equivalence classes, and then provide finer characteristics of their spectrum and accordingly the dynamics. 

\appendix
\section{Signed bipartite graphs}
\label{sec:app_spect-bi}
We now show more results regarding signed bipartite graphs. We know that every cycles in bipartite graphs have even length, and that if $\lambda$ is an eigenvalue of the (weighted) adjacency matrix, so is $-\lambda$. We show in Proposition \ref{pro:transition-spect-rho-bi} further characteristics of the spectral properties in the signed case. 
\begin{proposition}
	Suppose $\bar{G}$ is bipartite, or has period $2$, and we maintain the same notations of eigenvalues, eigenvectors and spectral radius as in Proposition \ref{pro:transition-spect-rho}. Let $V_{p1}, V_{p2}$ denote the corresponding bipartition for the bipartite structure, where $V_{p1}\cup V_{p2} = V,\, V_{p1}\cap V_{p2} = \emptyset$ and $E\subset V_{p1}\times V_{p2}$. If $G$ is balanced with the bipartition $V_{b1}, V_{b2}$, 
	\begin{enumerate}
		\item $G$ is also antibalanced with the bipartition $V_{a1}, V_{a2}$, where $V_{a1} = (V_{p1}\cap V_{b1})\cup(V_{p2}\cap V_{b2})$ and $V_{a2} = (V_{p1}\cap V_{b2})\cup(V_{p2}\cap V_{b1})$;
		\item $\lambda_1 = \rho(\mathbf{W}) > 0$, $\lambda_n = -\rho(\mathbf{W}) < 0$, and each eigenvalue is simple;
		\item $\mathbf{u}_1$ is the only eigenvector that has positive values in one node subset in the bipartition for the balanced structure (e.g., $V_{b1}$) and negative values in the other (e.g., $V_{b2}$), while $\mathbf{u}_n$ is the only eigenvector that has positive values in one node subset in the bipartition for the antibalanced structure (e.g., $V_{a1}$) and negative values in the other (e.g., $V_{a2}$).
	\end{enumerate}
	\label{pro:transition-spect-rho-bi}
\end{proposition}
\begin{proof}
	1. Since $G$ is balanced with the bipartition $V_{b1}, V_{b2}$, then if $W_{ij} > 0$, we have $v_i,v_j\in V_{bl},\, l\in\{1,2\}$, while if $W_{ij} < 0$, we have $v_i\in V_{b1},\, v_j\in V_{b2}$ or $v_i\in V_{b2},\, v_j\in V_{b1}$. Since $G$ is bipartite with the bipartition $V_{p1}, V_{p2}$, then for each edge $(v_i,v_j)$ with $v_i,v_j\in V_{a1} = (V_{p1}\cap V_{b1})\cup(V_{p2}\cap V_{b2})$,
	\begin{align*}
		(v_i,v_j) \in V_{a1}\times V_{a1} \Leftrightarrow (v_i,v_j) \in (V_{p1}\cap V_{b1})\times(V_{p2}\cap V_{b2}) \subset V_{b1}\times V_{b2},
	\end{align*}
	thus $W_{ij} < 0$. Similarly, we can show that for each edge $(v_i,v_j)$, (i) if $v_i,v_j\in V_{a2} = (V_{p1}\cap V_{b2})\cup(V_{p2}\cap V_{b1})$, $W_{ij} < 0$, while (ii) if $v_i\in V_{a1},\, v_j\in V_{a2}$ or $v_i\in V_{a2},\, v_j\in V_{a1}$, $W_{ij} > 0$. Hence, $G$ is also antibalanced with the bipartition $V_{a1}, V_{a2}$.
	
	2 and 3. Since $\bar{\mathbf{W}}$ is an non-negative matrix, and $\bar{G}$ is irreducible and has period $2$, then by Perron-Frobenius theorem, (i) $\rho(\bar{\mathbf{W}})$ is real positive and an eigenvalue of $\bar{\mathbf{W}}$, i.e., $\bar{\lambda}_1 = \rho(\bar{\mathbf{W}})$, (ii) this eigenvalue is simple s.t.~the associated eigenspace is one-dimensional, (iii) the associated eigenvector, i.e., $\bar{\mathbf{u}}_1$, has all positive entries and is the only one of this pattern, and (iv) $\bar{\mathbf{W}}$ has $2$ eigenvalues of the magnitude $\rho(\bar{\mathbf{W}})$. For bipartite graphs, we know that the other eigenvalue of magnitude $\rho(\bar{\mathbf{W}})$ is $\bar{\lambda}_n = -\rho(\bar{\mathbf{W}})$.
	
	Then, since $G$ is balanced, from Theorem \ref{the:transition-spect}, (i) $\mathbf{W}$ and $\bar{\mathbf{W}}$ share the same spectrum, and (ii) $\mathbf{U} = \mathbf{S}_{b1}\bar{\mathbf{U}}$, where $\mathbf{U} = [\mathbf{u}_1, \mathbf{u}_2, \dots, \mathbf{u}_n]$ and $\bar{\mathbf{U}} = [\bar{\mathbf{u}}_1, \bar{\mathbf{u}}_2, \dots, \bar{\mathbf{u}}_n]$ containing all the eigenvectors, and $\mathbf{S}_{b1}$ is the diagonal matrix whose $(i,i)$ element is $1$ if $v_i\in V_{b1}$ and $-1$ otherwise. Hence, $\lambda_1 = \bar{\lambda}_1 = \rho(\bar{\mathbf{W}}) = \rho(\mathbf{W})$, and this eigenvalue is simple. Meanwhile, $\mathbf{u}_1 = \mathbf{S}_{b1}\bar{\mathbf{u}}_1$, thus it has the pattern as described and is the only one of this pattern.
	
	Finally, since $G$ is also antibalanced, but with another bipartition $V_{a1}, V_{a2}$, from Theorem \ref{the:transition-spect}, (i) the spectrum of $\mathbf{W}$ is the same as negating that of $\bar{\mathbf{W}}$, and (ii) $\mathbf{U} = \mathbf{S}_{a1}\bar{\mathbf{U}}$, where $\mathbf{S}_{a1}$ is the diagonal matrix whose $(i,i)$ element is $1$ if $v_i\in V_{a1}$ and $-1$ otherwise. Hence similarly, $\lambda_n = -\bar{\lambda}_1 = -\rho(\bar{\mathbf{W}}) = -\rho(\mathbf{W})$, and this eigenvalue is simple. Meanwhile, $\mathbf{u}_n = \mathbf{S}_{a1}\bar{\mathbf{u}}_1$, thus it has the pattern as described and is the only one of this pattern.
\end{proof}

\section{Further details}
\label{sec:app_proofs}
In this section, we give the detailed proofs of some statements in the main text, and more results in the experiments in Figure \ref{fig:ssbm-linear-t}.
\begin{proof}[Constructive proof of Proposition \ref{pro:balanced-anti-bipart}]
	We consider an arbitrary signed tree graph, $T = (V, E, \mathbf{W})$. We first show that $T$ is balanced. For a node $v_i\in V$, there is only one path from $v_i$ to other nodes in $V$. Hence, we can partition each node $v_j\in V\backslash\{v_i\}$ according to the sign of the path from $v_i$ to $v_j$, where $V_1$ contains $v_i$ and the nodes of positive paths, while $V_2$ contains the nodes of negative paths. 
	
	We can show that $V_1, V_2$ is the bipartition corresponding to the balanced structure. Suppose there is an edge $(v_h,v_l)$ in $V_1$, then there are two paths of positive sign from $v_i$ to $v_h$ and $v_l$, and one of them does not go through edge $(v_h,v_l)$. WLOG, the path to $v_h$, denoted $P_h$, does not go through $(v_h,v_l)$. Then $P_h + (v_h,v_l)$ is a path from $v_i$ to $v_l$, thus positive. Hence, edge $(v_h,v_l)$ is positive. Similarly, we can show that each edge in $V_2$ is positive, while each edge between $S$ and $\bar{S}$ is negative. Hence, each signed tree is balanced. 
	
	We now show that $T$ is also antibalanced. We first construct another tree by negating the edge sign, $T' = (V, E, -\mathbf{W})$, and then following the above procedure, we can show that $T'$ is balanced. Hence, $T$ is antibalanced. 
\end{proof}

\begin{proof}[Detailed proof of Proposition \ref{pro:balanced-anti-bipart}]
	The balanced graphs can be equivalently defined as that all cycles have an even number of negative edges, and the antibalanced graphs can be equivalently defined as that all cycles have an even number of positive edges. Hence, a non-tree signed graph $G$ is both balanced and antibalanced if and only if every cycle has both an even number of positive edges and an even number of negative edges. This can only happen when there is no odd cycle, thus $G$ is bipartite. 
\end{proof}

\begin{proof}[Transition-matrix focused proof of Corollary \ref{cor:balance-lambda}]
	When $G$ is balanced, by Theorem \ref{the:transition-spect}, $\mathbf{P}_{sym}$ shares the same spectrum as $\bar{\mathbf{P}}_{sys}$, then $\lambda_1 = \bar{\lambda}_1 = 1$, and $\mathbf{P}$ also has eigenvalue $1$.
	
	We then consider the case when $\mathbf{P}$ has eigenvalue $1$, and show that $G$ is balanced. By Gershgorin circle theorem, we know that the spectral radius of $\mathbf{P}$ satisfies $\rho(\mathbf{P}) \le 1$, hence $1$ is the largest eigenvalue of $\mathbf{P}$, i.e., $\lambda_1 = 1$ which is shared by $\mathbf{P}_{sys}$. Then
	\begin{align*}
		&\mathbf{P}_{sys}\left(\mathbf{D}^{1/2}\mathbf{u}_1\right) = \mathbf{D}^{1/2}\mathbf{u}_1\\
		&\Leftrightarrow \left(\mathbf{D}^{1/2}\mathbf{u}_1\right)^T\mathbf{P}_{sys}\left(\mathbf{D}^{1/2}\mathbf{u}_1\right) = \left(\mathbf{D}^{1/2}\mathbf{u}_1\right)^T\mathbf{D}^{1/2}\mathbf{u}_1\\
		&\Leftrightarrow 2\sum_{(v_i,v_j)\in E}W_{ij}u_{1i}u_{1j} = \sum_{i}d_iu_{1i}^2 = \sum_i\left(\sum_{j}\abs{W_{ij}}\right)u_{1i}^2\\
		&\Leftrightarrow \sum_{(v_i,v_j)\in E}\abs{W_{ij}}\left(u_{1i} - sign(W_{ij})u_{1j}\right)^2 = 0\\
		&\Leftrightarrow 
		\begin{cases}
			u_{1i} = u_{1j},\quad W_{ij} > 0,\\
			u_{1i} = -u_{1j},\quad W_{ij} < 0,
		\end{cases}
	\end{align*}
	where $\mathbf{u}_1 = (u_{1i})$ and $sign(\cdot)$ returns the sign of the value. Then we note that such a vector exists if and only if we can find a bipartition $V_1, V_2$ of $V$ s.t.~all edges inside $V_1$ or $V_2$ are positive, while those between the two are negative, i.e., $G$ is balanced.
\end{proof}

\begin{figure}[ht]
	\centering
	\hspace*{-2em}
	\begin{tabular}{cc}
		\includegraphics[width=.48\textwidth]{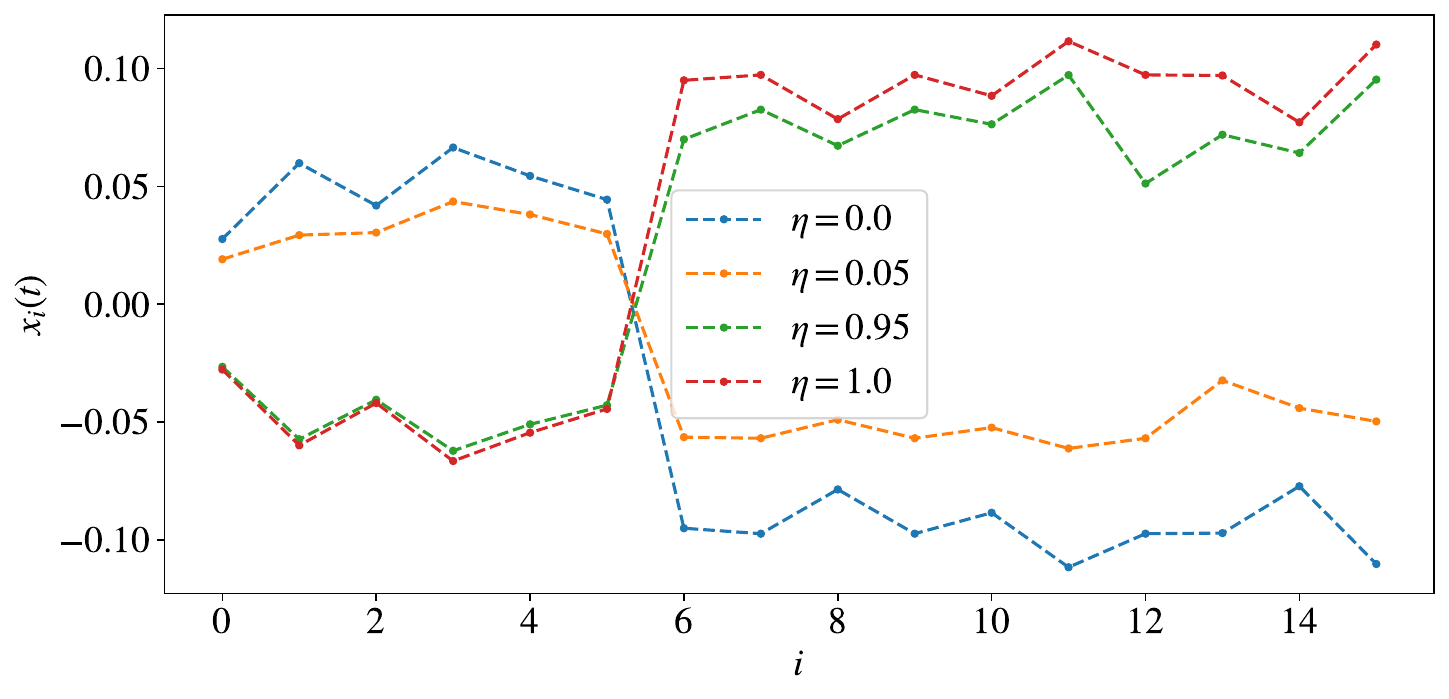} & \includegraphics[width=.48\textwidth]{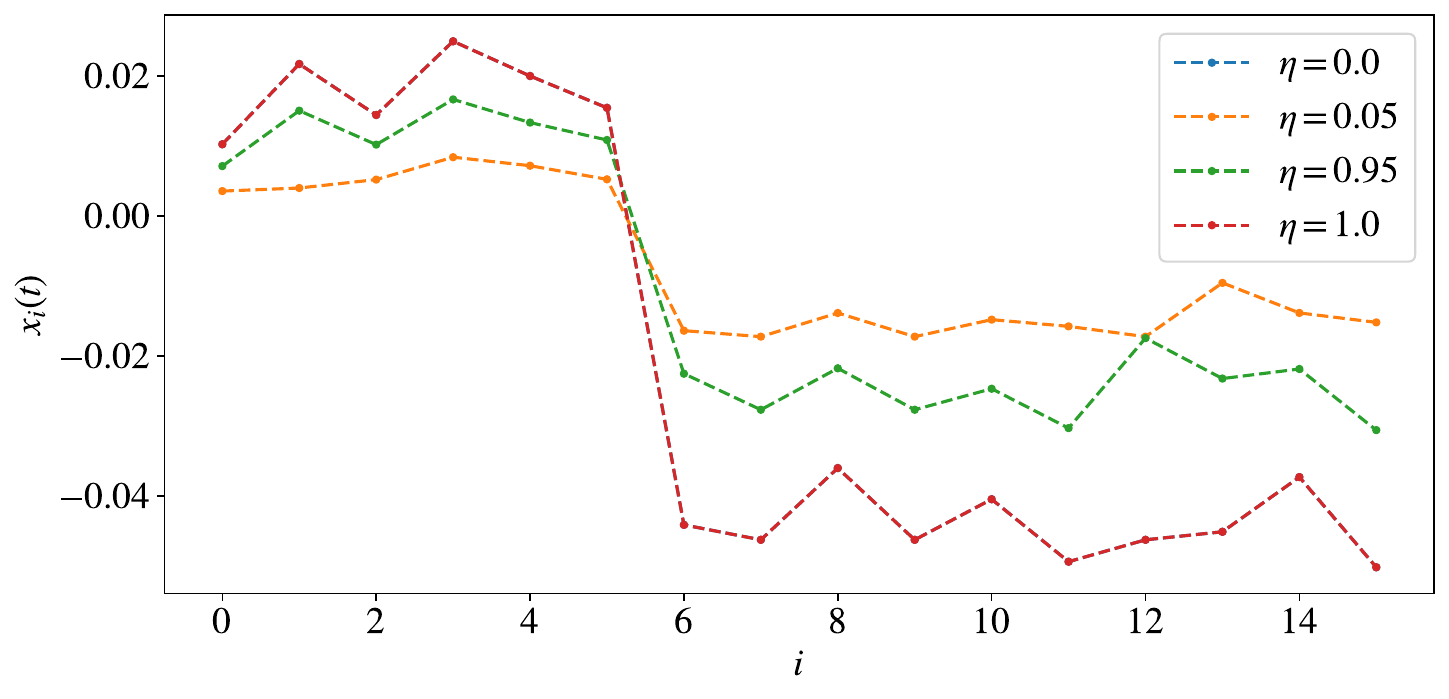} \\
		\includegraphics[width=.48\textwidth]{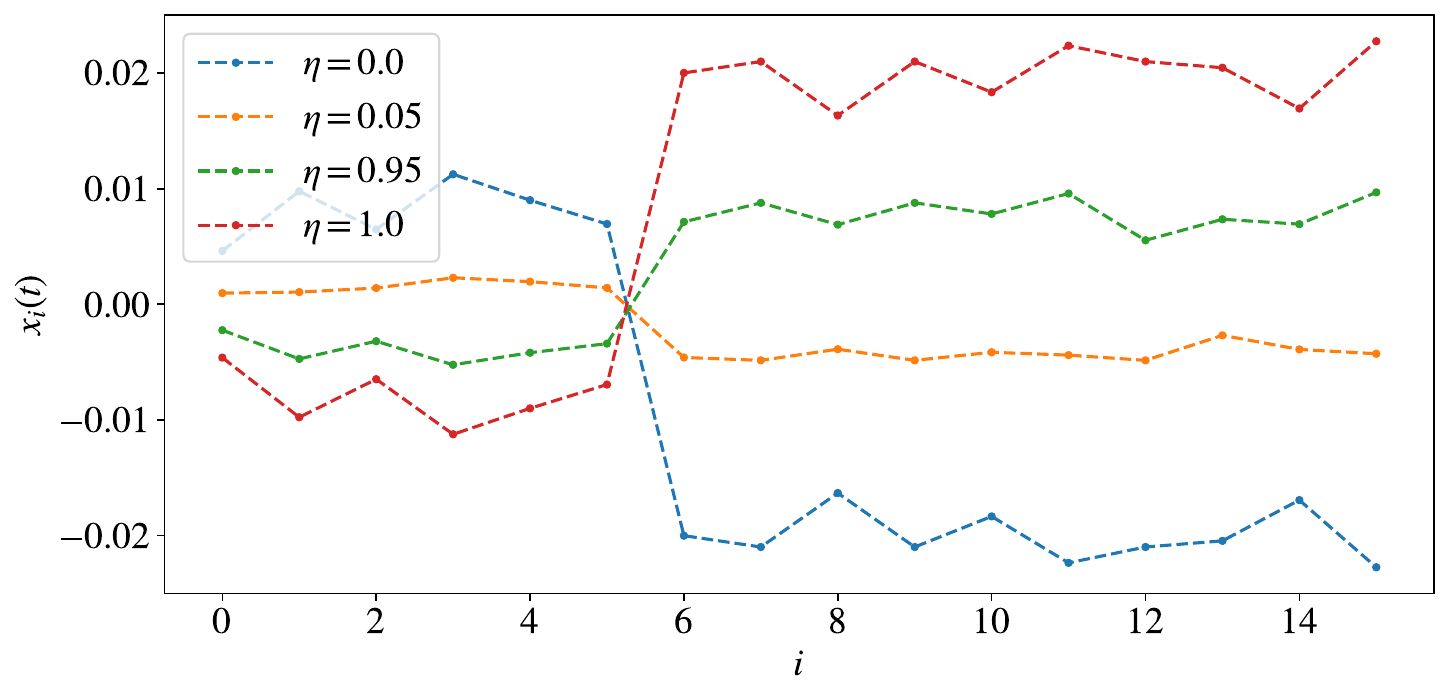} & \includegraphics[width=.48\textwidth]{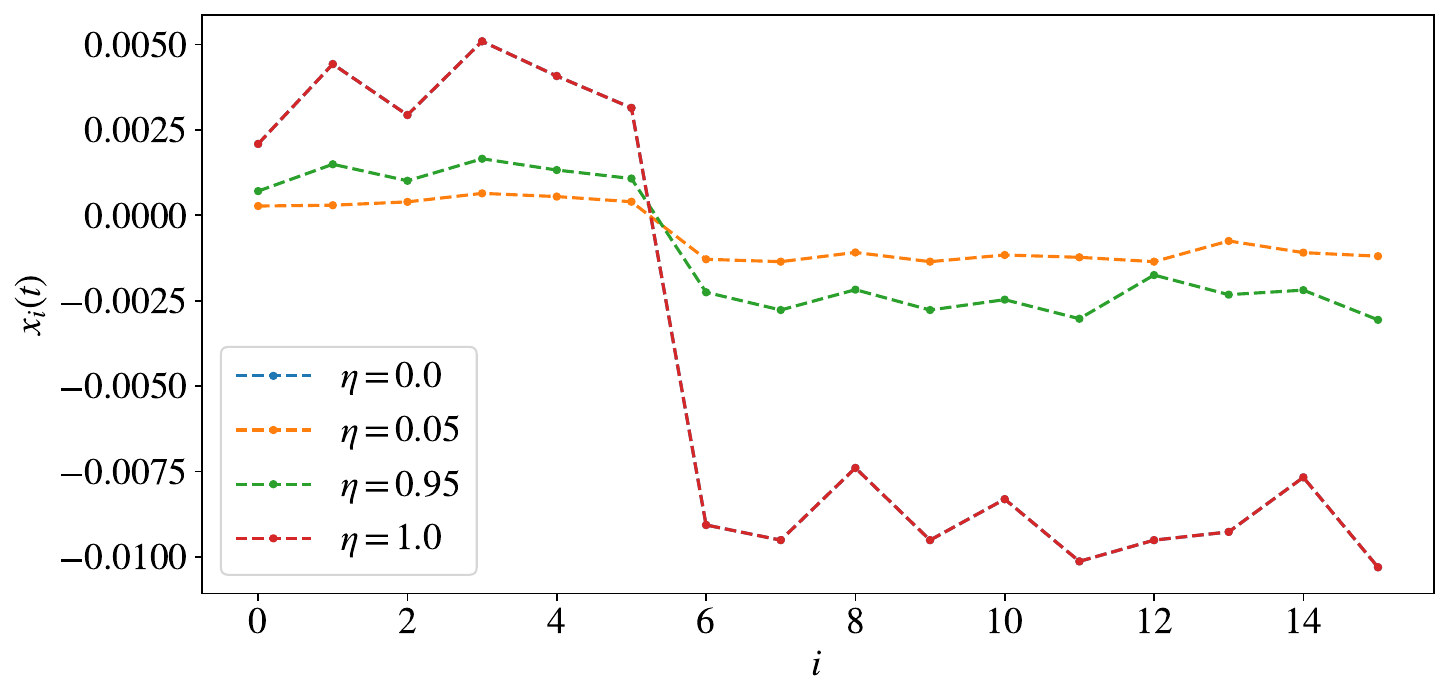}
	\end{tabular}
	\caption{State values from the linear adjacency dynamics on the networks in Figure \ref{fig:exp-ssbms} that are balanced ($\eta = 0$, blue), close to being balanced ($\eta = 0.05$, orange), antibalanced ($\eta=1$, red), and close to being antibalanced ($\eta = 0.95$, green), when $t=5$ (upper left), $10$ (upper right), $15$ (bottom left) and $20$ (bottom right).}
	\label{fig:ssbm-linear-t}
\end{figure}

\section*{Acknowledgments}
%The work was partially done while Y.T.~was at Oxford, where she was funded by the EPSRC Centre for Doctoral Training in Industrially Focused Mathematical Modelling (EP/L015803/1) in collaboration with Tesco PLC.
We thank Gesine Reinert and Ginestra Bianconi for useful discussions and suggestions. Y.T.~is funded by the Wallenberg Initiative on Networks and Quantum Information (WINQ). The work was partially done when Y.T.~was at Mathematical Institute, University of Oxford, where she was funded by the EPSRC Centre for Doctoral Training in Industrially Focused Mathematical Modelling (EP/L015803/1) in collaboration with Tesco PLC. R.L.~acknowledges support from the EPSRC Grants EP/V013068/1 and EP/V03474X/1. 

\bibliographystyle{plain}
\bibliography{references}

%%%%%%%%%%%%%%%%%%
%\clearpage
%\appendix

\end{document}